\documentclass[a4paper,10pt]{article}
\usepackage{a4wide}
\usepackage[ansinew]{inputenc}
\usepackage[spanish,english]{babel}
\usepackage{amssymb}
\usepackage{amsmath}
\usepackage{graphicx}
\usepackage[usenames,dvipsnames]{color}
\usepackage{color}
\usepackage{lineno}

\newtheorem{definition}{Definition}
\newtheorem{theorem}{Theorem}

\newtheorem{lemma}{Lemma}
\newtheorem{corollary}{Corollary}

\newtheorem{property}{Property}
\newtheorem{problem}{Problem}

\newtheorem{MIrem}{Remark}
\newenvironment{remark}    {\begin{MIrem}\em}    {\end{MIrem}}

\newenvironment{proof}{\noindent\textbf{Proof.}}{\hfill \ensuremath {\boxtimes}\\}

\newcommand{\ch}{\textrm{CH}}

\newcommand{\rr}{\texttt{r}}
\newcommand{\br}{\texttt{b}}

\date{}

\title{On Hamiltonian alternating cycles and paths\footnote{The authors dedicate this paper to the memory of Prof. Ferran Hurtado who proposed the relaxation from being plane to being $1$-plane for cycles and paths on bicolored point sets. It was a great pleasure to work with him.}}

\author{\normalfont\fontsize{12}{14}\selectfont
Merc\`e Claverol$^1$\and
Alfredo Garc\'{i}a$^2$\and
Delia Garijo$^3$\and
Carlos Seara$^1$\and
Javier Tejel$^2$\and
}

\begin{document}

\maketitle

\addtocounter{footnote}{1} \footnotetext{Departament de Matem\`atiques, Universitat Polit\`ecnica de
Catalunya, Spain. {\bf merce.claverol@upc.edu} and {\bf carlos.seara@upc.edu}.
Partially supported by projects Gen. Cat. DGR 2014SGR46 and MTM2015-63791-R MINECO/FEDER.}
\addtocounter{footnote}{1} \footnotetext{Departamento de M\'etodos Estad\'{i}sticos, IUMA, Universidad de
Zaragoza, Spain. {\bf olaverri@unizar.es} and {\bf jtejel@unizar.es}.
Partially supported by projects E58 (ESF) (Gobierno de Arag\'on) and MTM2015-63791-R MINECO/FEDER.}
\addtocounter{footnote}{1} \footnotetext{Departamento de Matem\'atica Aplicada I, Universidad de Sevilla, Spain. {\bf dgarijo@us.es}. Partially supported by projects MTM2014-60127-P and FQM-164.}

\begin{abstract}
We undertake a study on computing Hamiltonian alternating cycles and paths on bicolored point sets. This has been an intensively studied problem, not always with a solution, when the paths and cycles are also required to be plane. In this paper, we relax the constraint on the cycles and paths from being plane to being 1-plane, and deal with the same type of questions as those for the plane case, obtaining a remarkable variety of results. For point sets in general position, our main result is that it is always possible to obtain a $1$-plane Hamiltonian alternating cycle. When the point set is in convex position, we prove that every Hamiltonian alternating cycle with minimum number of crossings is $1$-plane, and provide $O(n)$ and $O(n^2)$ time algorithms for computing, respectively, Hamiltonian alternating cycles and paths with minimum number of crossings.
\end{abstract}

\section{Introduction}\label{sec:intro}

A \emph{geometric graph} is a graph drawn in the plane whose vertex set is a set of points and whose edges are straight-line segments connecting pairs of points. Two edges of a geometric graph \emph{cross} if they have an intersection point lying in the relative interior of both edges. When there are no crossings, the geometric graph is said to be \emph{plane}, and it is \emph{alternating} if its vertex set is a bicolored point set and all its edges are \emph{bichromatic}, i.e., they connect points of different colors; when all edges connect points of the same color, the geometric graph is \emph{monochromatic}. In this paper, we deal with alternating (geometric) cycles and paths.

The study of alternating paths was initiated, according to Pach~\cite{KPT04}, by Erd\H{o}s~\cite{PH} around 1989, who proposed the following problem:

\vspace{0.1cm}
\emph{Determine or estimate the largest number $\ell(n)$ such that, for every set of $n$ red and $n$ blue points on a circle, there exists a non-crossing alternating
path consisting of $\ell(n)$ vertices.}
\vspace{0.1cm}

This problem comes from the fact that it is not always possible to obtain a plane Hamiltonian alternating path on that point configuration. Erd\H{o}s conjectured that $\ell(n)= \frac{3}{2}n+2+o(n)$, but one can find point configurations for which  $\ell(n)<\frac{4}{3}n +o(n)$, as described in~\cite{AGHT03,KPT04,M11}.
The best bounds up to date for $\ell(n)$ are due to Kyn\v{c}l, Pach and T\'{o}th~\cite{KPT04}, and valid for bicolored point sets in general position (i.e, no three points lie on the same line). However, the conjecture that $|\ell(n)-\frac{4}{3}n|= o(n)$ remains open, even for points in convex position. For more information on this topic see~\cite{AGHNR99,CKMSV09a}.

In this context, Akiyama and Urrutia~\cite{AU90} gave an $O(n^2)$ time algorithm for computing a plane alternating path visiting the maximum possible number of points, provided that the point set is in convex position and the endpoints of the path are previously fixed. Other references on bicolored point sets are~\cite{KK03} (which is the book of reference) and~\cite{DK01,DS00,T96} (where monochromatic paths and matchings are studied). See also~\cite{CKMSV09a} for results on alternating paths on double chains, and~\cite{AGHNR99} for results on plane alternating spanning tress. For other problems on geometric graphs see~\cite{BMP05,HT13,P13}.

Since the plane character for Hamiltonian alternating cycles and paths is not always possible to achieve, it is natural to relax the constraint in order to obtain better results. Thus, different approaches arise based on allowing crossings. Kaneko, Kano and Yoshimoto~\cite{KKY00} focused in getting few crossings in global, and proved that there always exists a Hamiltonian alternating cycle on a bicolored point set in general position with at most $n-1$ crossings, where $n$ is the number of points of each color. Claverol et al.~\cite{CGHLS13} studied the \emph{1-plane} character, i.e., every edge is allowed to have at most one crossing. They proved that one can always obtain a $1$-plane Hamiltonian alternating cycle on a point set in convex position and on a double chain. They also conjectured the existence of a $1$-plane Hamiltonian alternating path for point sets in general position.  Observe that the terms plane graph and $1$-plane graph refer to a geometric object, while to be planar or $1$-planar are properties of the underlying abstract graph; see~\cite{R65} where the concept of 1-planar graph was introduced (in relation to a coloring problem for planar graphs).

Our paper can be considered as a continuation of the study developed by Claverol et al.~\cite{CGHLS13} since we explore the relaxation from being plane to being $1$-plane for paths and cycles on bicolored point sets. To be more precise on the content of this paper, we next provide some notation and additional definitions.

Let $R$ and $B$ be two disjoint sets of red and blue points in the plane, respectively, such that $S=R\cup B$ is in general position\footnote{In all the figures in this paper, red points are illustrated as solid red points, and blue points are depicted as hollow blue points.}.
Unless otherwise stated, $|B|=n$ and $n\leq |R|\leq n+1$. For short, we shall use $1$-PHAC and  $1$-PHAP to refer to a $1$-plane Hamiltonian alternating cycle and a $1$-plane Hamiltonian alternating path, respectively. As usual, $\ch(X)$ denotes the convex hull of a point set $X\subseteq S$. We say that $\ch(X)$ is \emph{bichromatic} if its boundary contains points of both colors, and \emph{monochromatic} otherwise.
A \emph{run} of $X\subseteq S$ is a maximal set of consecutive points of the same color on the boundary of $\ch(X)$; it can be \emph{red} or \emph{blue} depending on the color of the points, and its \emph{cardinality} is the number of its points. Let $\rr(X)$ and \br$(X)$ be the number of red and blue runs of $X$, respectively. Note that $\rr(X)=\br(X)$ whenever $\ch(X)$ is bichromatic; otherwise one of the values equals $1$ and the other zero.

Our main result in Section~\ref{sec:general} is that a $1$-PHAC can always be obtained on all instances of set $S$ in general position, also upper-bounding the number of crossings by $n-\max \{\rr(S), \br(S)\}$. The key tool to prove this result is that, except for some special configurations, one can always draw a $1$-PHAP on $S$, fixing previously its endpoints on the boundary of $\ch(S)$. We conclude the section by showing that a $1$-PHAC on $S$ can be computed in $O(n^2)$ time and space.

Section~\ref{sec:convex} concerns Hamiltonian alternating cycles and paths on point sets $S$ in convex position.
Subsection~\ref{subsec:cycles} deals with cycles:  $n-\rr(S)$
is shown to be a lower bound for the number of crossings, and only some cycles that are 1-plane attain the bound. We also  prove that a $1$-PHAC on $S$ with minimum number of crossings can be computed in $O(n)$ time and space, and as a consequence of the process for cycles, the same complexity is obtained for computing a  $1$-PHAP  with minimum number of crossings, provided that the endpoints of the path are consecutive points of $S$.

Subsection~\ref{subsec:paths} contains our results for paths on sets in convex position. We first prove that the above-mentioned special configurations are the unique point configurations that do not admit a $1$-PHAP with given endpoints. Then, we show that, except for those special configurations, every Hamiltonian alternating path with minimum number of crossings is 1-plane. This result allows us to design an $O(n^2)$ time and space algorithm to compute a $1$-PHAP with fixed endpoints and minimum number of crossings.

We conclude the paper in Section~\ref{sec:conclu} with some open problems.

\section{$1$-PHAC for point sets in general position}\label{sec:general}

Our aim in this section is to prove that there always exists a $1$-PHAC on a set $S=R\cup B$ in general position with at most $n-\max \{\rr(S), \br(S)\}$ crossings, where $n=|R|=|B|\geq 2$ (Figure~\ref{fig:cycle} left shows an example); see Theorem~\ref{the:cycle} below.

The proof of Theorem~\ref{the:cycle} follows that given by Kaneko, Kano and Yoshimoto~\cite{KKY00} to construct a Hamiltonian alternating cycle with at most $n-1$ crossings. Although both proofs are similar, there are two fundamental differences. On the one hand, the construction of the cycle in~\cite{KKY00} is based on the fact that one can always draw a Hamiltonian alternating path with at most $n-1$ crossings connecting any red point to any blue point, both located on the boundary of $\ch(S)$. Nevertheless, this is not true for 1-plane geometric graphs (see Figure~\ref{fig:cycle} right) and so our proof is adapted to this new scenario. On the other hand, our proof makes possible upper bounding the number of crossings, not only by $n-1$ as in~\cite{KKY00}, but using the number of runs of $S$.

\begin{figure}[tb]
\centering
\includegraphics[scale=0.7]{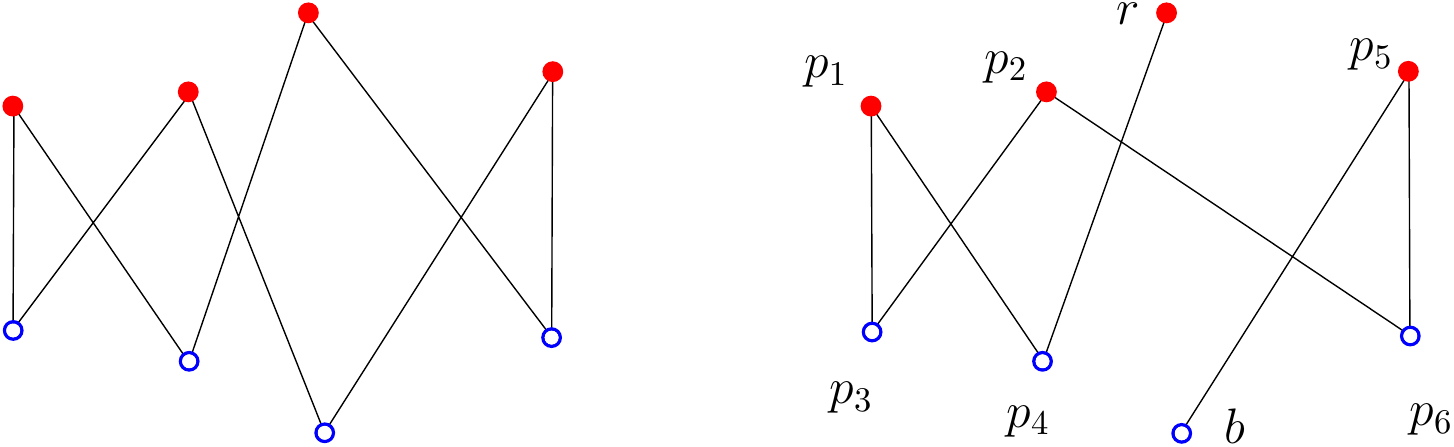}
\caption{Left: A $1$-PHAC. Right: A point configuration that admits no $1$-PHAP with endpoints $r$ and $b$.}\label{fig:cycle}
\end{figure}

The following technical lemma, which can be viewed as an adaptation of Lemma 1 from Kaneko, Kano and Yoshimoto~\cite{KKY00}, is crucial to reach some of the main results in this paper.

\begin{lemma}\label{lemm:partition}
Let $S=R\cup B$ be in general position, and let $r$ be an arbitrary red point of $S$. Then, the following statements hold.
\begin{itemize}
  \item[(i)] Suppose that $|R|=|B|=n\geq 1$, $\{p_1,\ldots,p_{2n-1}\}$ is the counterclockwise radial order of the points of $S\setminus\{r\}$ around $r$, and $p_{i+1}$ is the first point in that sequence such that $$|\{r,p_1,\ldots,p_{i+1}\}\cap R|=|\{r,p_1,\ldots,p_{i+1}\}\cap B|.$$ Then, (a) $p_{i+1}=p_1$ if $p_1$ is blue, (b) $p_i$ and $p_{i+1}$ both exist and are blue whenever $p_1$ is red. In this last case, we say that $S_1=\{p_1,\ldots,p_i\}$ and $S_2=\{p_{i+1},\ldots,p_{2n-1}\}$ form a \emph{counterclockwise partition of $S\setminus \{r\}$ around $r$}; see Figure~\ref{fig:partition} left.
  \item[(ii)] Suppose that $|R|=|B|+1=n+1\geq 2$, $\{p_1,\ldots,p_{2n}\}$ is the counterclockwise radial order of the points of $S\setminus\{r\}$ around $r$, point $p_{2n}$ is red, and $p_{i+1}$ is the first point in that sequence such that $$|\{r,p_1,\ldots,p_{i+1}\}\cap R|=|\{r,p_1,\ldots,p_{i+1}\}\cap B|.$$ Then, (a) $p_{i+1}=p_1$ if $p_1$ is blue, (b) $p_i$ and $p_{i+1}$ both exist and are blue whenever $p_1$ is red. In this last case, we say that $S_1=\{p_1,\ldots,p_i\}$ and $S_2=\{p_{i+1},\ldots,p_{2n}\}$ form a \emph{counterclockwise partition of $S\setminus\{r\}$ around $r$}.
\end{itemize}
\end{lemma}

\begin{proof}
We prove statement $(i)$. An analogous argument, by removing point $p_{2n}$, proves statement $(ii)$. If $p_1$ is blue, the result is obvious. If $p_1$ is red, let $\Delta_i=|\{r,p_1,\ldots,p_i\}\cap R|-|\{r,p_1,\ldots,p_i\}\cap B|$ for $1\leq i\leq 2n-1$. The sequence of values of $\Delta_i$ begins with value $2$, ends with value $0$,  increases by $1$ if the visited point is red, and decreases by $1$ if the visited point is blue. Therefore, when the value $0$ appears the first time in that sequence, it is produced by a subsequence $\ldots,2,1,0$, which implies that the last two visited points must be blue. See Figure~\ref{fig:partition} left.
\end{proof}

A totally symmetric result holds if we explore the points of $S\setminus\{r\}$ in clockwise order: if $\{p'_1,\ldots,p'_{2n-1}\}$ ($\{p'_1,\ldots,p'_{2n}\}$) is the clockwise radial order of the points of $S\setminus\{r\}$ around $r$ and point $p'_1$ is red, then there is a first clockwise point $p'_{i'+1}$ such that $$|\{r,p'_1,\ldots,p'_{i'+1}\}\cap R|=|\{r,p'_1,\ldots,p'_{i'+1}\}\cap B|$$ and points $p'_{i'}$, $p'_{i'+1}$ are blue. We say that $S_1=\{p'_1,\ldots,p'_{i'}\}$ and $S_2=\{p'_{i'+1},\ldots,p'_{2n-1}\}$ ($S_2=\{p'_{i'+1},\ldots,p'_{2n}\}$) form a \emph{clockwise partition of $S\setminus\{r\}$ around $r$}. See Figure~\ref{fig:partition} right.

Obviously, clockwise and counterclockwise partitions of $S\setminus\{b\}$ around a blue point $b$ can also be defined when $|R|=|B|=n\geq 1$ or when $|B|=|R|+1=n+1\geq 2$ and $p_{2n}$ is blue.

\begin{paragraph}
\noindent \textbf{Notation according to Lemma~\ref{lemm:partition}.} Hereafter, given a partition $S_1\cup S_2$ of $S\setminus\{r\}$ around $r$,
we shall use  $n_1=|R\cap S_1|=|B\cap S_1|$ and $n_2=|R\cap S_2|=|B\cap S_2|$ when $|R|=|B|+1$. For $|R|=|B|$ it follows that $|B\cap S_2|= n_2+1$ and $|R\cap S_2|= n_2$. Further, we shall write $P_{X}(p,q)$ to indicate a $1$-PHAP on a set $X\subseteq S$ with endpoints $p,q\in X$ (here, it is understood that $p$ and $q$ have adequate colors, and that $X$ contains a proper number of red and blue points in order to construct such a path).
\end{paragraph}

\begin{figure}[tb]
\centering
\includegraphics[scale=0.7]{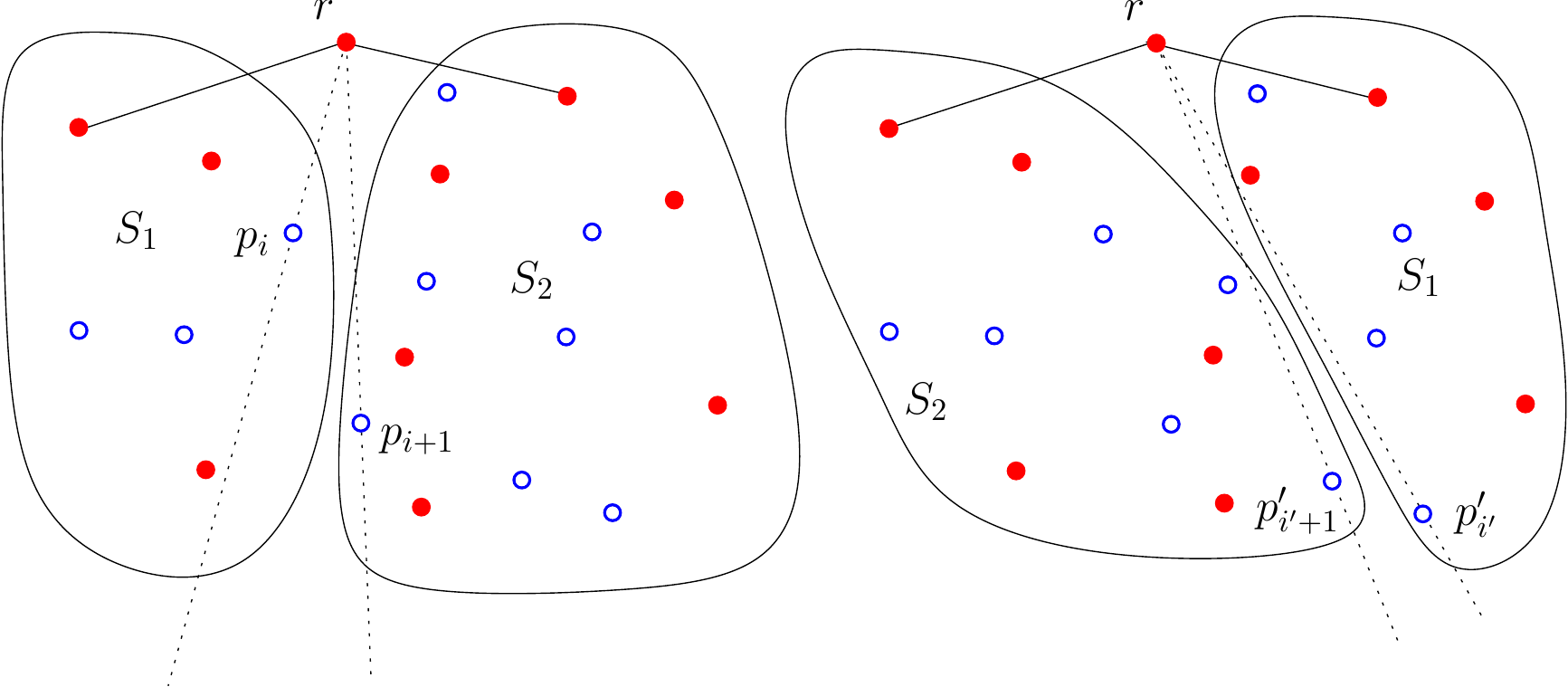}
\caption{A counterclockwise partition (left) and a clockwise partition (right) of $S$ around $r$.}\label{fig:partition}
\end{figure}

\vspace{0.5cm}

Lemma~\ref{lemm:partition} and the notion of \emph{visibility} provided below are the key tools to prove Lemma \ref{lemm:path}, which states that a $1$-PHAP connecting any two points on the boundary of $\ch(S)$ can always be drawn except for the so-called  \emph{special configurations} (see Definition~\ref{def:special}),  also upper-bounding the number of crossings in terms of the number of runs of $S$. This is the main result to reach Theorem~\ref{the:cycle}.

For a point set $X$ (not necessarily consisting of red and blue points) and a point $x$ outside $\ch(X)$, a point $p\in X$ is {\it visible} from $x$ or point $x$ {\it sees} point $p$ if the segment $xp$ does not cross $\ch(X)$. Note that if $p$ and $p'$ are both visible from $x$, and $p$ is placed counterclockwise before $p'$ on the boundary of $\ch(X)$, then all points on the boundary of $\ch(X)$ placed clockwise between $p'$ and $p$ are also visible from $x$.

\begin{definition}\label{def:special}
Let $S=R\cup B$, and let $r\in R$ and $b\in B$ be two arbitrary points on the boundary of $\ch(S)$. The triple $(S,r,b)$ is a {\rm special configuration} if the  three following conditions are satisfied:
\begin{itemize}
\item[(i)] The two neighbors of $r$ on the boundary of $\ch(S)$ are red, and those of $b$ are blue.
\item[(ii)] $p_{i+1}=p'_{i'+1}=b$ in the clockwise and counterclockwise partitions of $S$ around $r$.
\item[(iii)] $p_{i+1}=p'_{i'+1}=r$ in the clockwise and counterclockwise partitions of $S$ around $b$.
\end{itemize}
Otherwise, the configuration $(S,r,b)$ is {\rm non-special}.
\end{definition}

Figure~\ref{fig:cycle} (right) illustrates a special configuration $(S,r,b)$. In the counterclockwise partition of $S$ around $r$, $S_1=\{p_1, p_2, p_3, p_4\}$ and so $p_{i+1} =b$; for the clockwise partition we have $S_1=\{p_5, p_6\}$, and again $p'_{i'+1} =b$. The same happens for the clockwise and counterclockwise partitions of $S$ around $b$: points $p_{i+1}$ and $p'_{i'+1}$ coincide with $r$.

\begin{lemma}\label{lemm:path}
Let $S=R\cup B$, and let $r$ be an arbitrary red point on the boundary of $CH(S)$. Then, the following statements hold.
\begin{itemize}
\item[(i)] If $|R|=|B|=n\geq 1$ and there is a blue point $b$ on the boundary of $\ch(S)$ such that the configuration $(S,r,b)$ is non-special, then there exists a $1$-PHAP on $S$ with endpoints $r$ and $b$ and at most $n-\rr(S)=n-\br(S)$ crossings.

 \item[(ii)] If $|R|=|B|+1 = n+1\geq 2$ and there is a red point $r'\in R\setminus\{r\}$ on the boundary of $\ch(S)$, then there exists a $1$-PHAP  on $S$  with endpoints $r$ and $r'$ and  at most $n-\br(S)$ crossings.
\end{itemize}
\end{lemma}

\begin{proof}
We prove simultaneously statements $(i)$ and $(ii)$ by induction on $|S|$. Note that one of the differences between the two statements is the color of the endpoints of the corresponding paths, they have either distinct color (red-blue) or the same color (red-red); obviously one can also consider the blue version of statement $(ii)$ for $|B|=|R|+1$ and two distinct blue points on the boundary of $\ch(S)$.
Thus, the three possible combinations of color will be used along the proof to construct, by induction, adequate 1-plane paths on subsets of $S$ and connect them in order to obtain a 1-plane path on $S$. Assume that $|S|>3$ (the result is straightforward for  $|S|=2$ or $|S|=3$).

\vspace{0.3cm}
\noindent{\bf Proof of statement $(i)$}
\vspace{0.3cm}

Suppose first that there exists a blue point $b'\ne b$ on the boundary of $\ch(S\setminus \{r\})$ that is visible from $r$. Then, by induction, there is a $1$-plane path $P_{S\setminus \{r\}}(b',b)$ with at most $n-1-\rr(S\setminus \{r\})$ crossings. Clearly, $\rr(S\setminus \{r\}) \ge \rr(S)$ except if $\{r\}$ is a red run that $\rr({S\setminus \{r\}}) = \rr(S)-1$. Thus, the path $P_{S}(r,b)=\{rb'\}\cup P_{S\setminus \{r\}}(b',b)$ connects $r$ and $b$ and has at most $n-1-(\rr(S)-1)=n-\rr(S)$ crossings. Since edge $rb'$ does not intersect $\ch(S\setminus \{r\})$, $P_{S}(r,b)$ is the desired $1$-plane path.

The preceding argument also applies if there is a red point $r'\ne r$ on the boundary of $\ch(S\setminus \{b\})$ that is visible from $b$. Thus, we  assume the following \emph{visibility property}: all points on the boundary of $\ch(S\setminus \{r\})$ visible from $r$ are red and all points on the boundary of $\ch(S\setminus \{b\})$ visible from $b$ are blue except for, respectively, $b$ and $r$  if they are consecutive on the boundary of $\ch(S)$.
We now distinguish two cases according to the position of $r$ and $b$.

\vspace{0.3cm}
\noindent{\emph{Case 1: points $r$ and $b$ are consecutive on the boundary of $\ch(S)$.}}
\vspace{0.3cm}

From the fact that $r$ and $b$ are consecutive and the visibility property, it follows that there exist a red point $r'$ and a blue point $b'$ consecutive on $\ch(S\setminus \{r,b\})$ and visible from both $r$ and $b$ (see Figure~\ref{fig:consecutive} left). Since condition $(i)$ of Definition~\ref{def:special} is not satisfied, the configuration $(S\setminus \{r,b\},b',r')$ is non-special  and so, by induction, there is a $1$-plane path $P_{S\setminus \{r,b\}}(b',r')$ with at most $n-1-\rr({S\setminus \{r,b\}})$ crossings.
Edges $rb'$ and $r'b$ do not intersect that path as their endpoints are visible, which implies that $P_{S}(r,b) = \{rb'\}\cup P_{S\setminus \{r,b\}}(b',r')\cup \{r'b\}$ is a $1$-PHAP on $S$ with endpoints $r$ and $ b$.

The red run of $S$ containing point $r$ has more than one point (note that $r$ only sees $b$ and red points) and so it remains in $S\setminus\{r,b\}$; an analogous argument follows for $b$. Therefore, $\rr({S\setminus \{r,b\}})\ge \rr(S)$ which upper bounds the number of crossings in $P_{S}(r,b)$ by $n-1-\rr(S)+1=n-\rr(S)$ (edges $rb'$ and $r'b$ might intersect).

\begin{figure}[!htbp]
\centering
\includegraphics[scale=0.7]{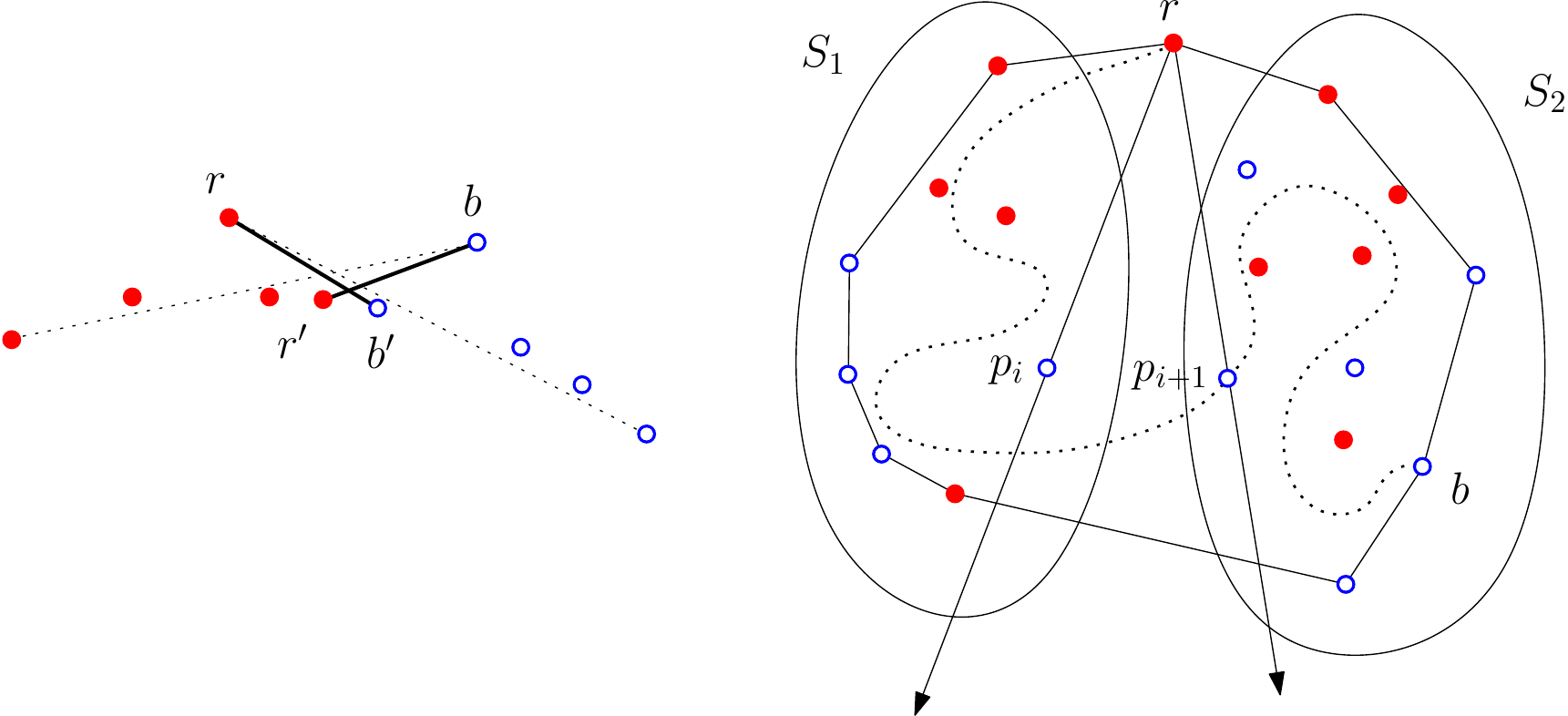}
\caption{Left: $r$ and $b$ are consecutive points on the boundary of $\ch(S)$. Right: $r$ and $b$ are not consecutive points.}\label{fig:consecutive}
\end{figure}

\vspace{0.3cm}
\noindent{\emph{Case 2: points $r$ and $b$ are non-consecutive on the boundary of $\ch(S)$.}}
\vspace{0.3cm}

By the visibility property, the two neighbors of $r$ on the boundary of $\ch(S)$ are red points (blue points for $b$), and so we can apply Lemma~\ref{lemm:partition} around $r$ or $b$ to partition $S$. Since the configuration $(S,r,b)$ is non-special, condition $(ii)$ or $(iii)$ of Definition~\ref{def:special} is not satisfied for some partition; assume, without loss of generality, that $p_{i+1}\neq b$ in the counterclockwise partition of $S$ around $r$. Recall that $p_i$ and $p_{i+1}$ are consecutive blue points in the counterclockwise radial order around $r$, as specified in Lemma~\ref{lemm:partition}.

Suppose that $b\in S_2$ (a symmetric construction can be done for $b\in S_1$); see Figure~\ref{fig:consecutive} right. By induction, there are 1-plane paths $P_{S_1\cup \{r, p_{i+1}\}}(r, p_{i+1})$ and $P_{S_2}(p_{i+1},b)$ with at most $n_1+1-\rr(S_1\cup \{r, p_{i+1}\})$ and $n_2-\rr(S_2)$ crossings, respectively. Connecting both 1-plane paths, we obtain a path $P_S(r,b)$ that is also 1-plane since $\ch(S_1\cup \{r, p_{i+1}\})$ and $\ch(S_2)$ do not intersect. Further, $n=n_1+n_2+1$ and so it  has at most $n-\rr(S_1\cup \{r, p_{i+1}\})-\rr(S_2)$ crossings.

Each red run of $S$ is also a red run of  $S_1\cup \{r, p_{i+1}\}$ or a red run of $S_2$, and the red run of $S$ containing point $r$ is counted twice, one in $S_1\cup \{r, p_{i+1}\}$ and the other in $S_2$. Thus $\rr(S_1\cup \{r, p_{i+1}\})+\rr(S_2)\geq \rr(S)+1$, which upper bounds the number of crossings in $P_S(r,b)$ by $n-\rr(S)-1$.

\vspace{0.3cm}
\noindent{\bf Proof of statement $(ii)$}
\vspace{0.3cm}

Let $C_r$ be the set of points of $\ch(S\setminus \{r\})$ that are visible from $r$, and suppose first that it contains red and blue points. Then, there are two consecutive points $b\in B$ and $r''\in R$ on the boundary of $\ch(S\setminus \{r\})$, and so the configuration $(S\setminus \{r\}, b,r')$ is  non-special (point $b$ has a red neighbor). Hence, by induction, one can construct a $1$-plane path $P_{S\setminus \{r\}}(b,r')$ with at most $n-\rr(S\setminus \{r\})=n-\br(S\setminus \{r\})$ crossings.

Since points $r$ and $b$ are visible, we can conclude that the path $P_{S}(r,r')=\{rb\}\cup P_{S\setminus \{r\}}(b,r')$ is a $1$-PHAP on $S$ with endpoints $r,r'$ and at most $n-\br(S\setminus \{r\}) \le n-\br(S)$ crossings (this inequality comes from the fact that $C_r$ contains points of both colors).

\begin{figure}[!htbp]
\centering
\includegraphics[scale=0.7]{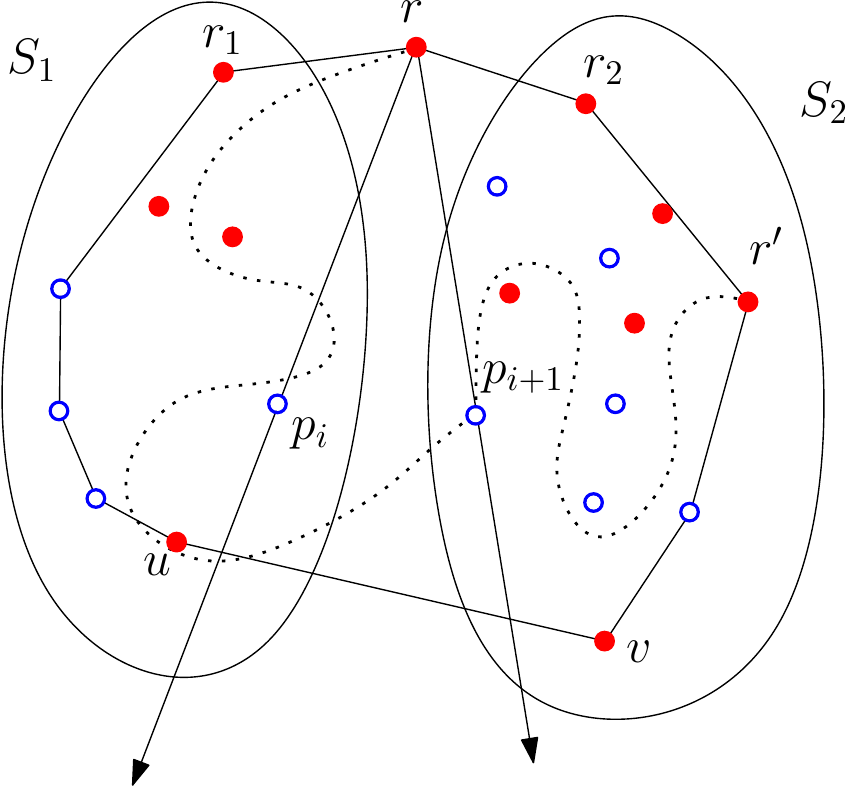}
\caption{The configuration $(S_2, p_{i+1}, r')$ is non-special.}\label{fig:subcase31}
\end{figure}

The preceding argument also applies if the set $C_{r'}$ of points of $\ch(S\setminus \{r'\})$ visible from $r'$ contains red and blue points. Therefore, we may assume that $C_r$ and $C_{r'}$ consist of points of the same color, and distinguish two cases according to it.

\vspace{0.3cm}
\noindent{\emph{Case 3: $C_r$ or $C_{r'}$ or both consist of only red points.}}
\vspace{0.3cm}

Suppose, without loss of generality, that $C_r$ consists of red points. Let $r_1$ and $r_2$ be the neighbors of $r$ on the boundary of $\ch(S)$, clockwise. Since $r_1$ is a red point, there is a counterclockwise partition $S_1\cup S_2$ of $S\setminus \{r\}$ around $r$ as described in Lemma~\ref{lemm:partition}. Moreover, points $p_i$ and $p_{i+1}$ (also specified in Lemma~\ref{lemm:partition}) are consecutive in the radial order around $r$, and so rays $rp_i$ and $rp_{i+1}$ cross the same edge $uv$ of the boundary of $\ch(S)$, with $u\in S_1$ and $v\in S_2$. Assume that $r'\in S_2$ (analogous for $r'\in S_1$), and consider two cases depending on whether $(S_2,p_{i+1},r')$ is a special configuration or not.

\vspace{0.3cm}
\noindent{\emph{Case 3.1: the configuration $(S_2,p_{i+1}, r')$ is non-special.}}
\vspace{0.3cm}

By induction, there are $1$-plane paths $P_{S_2}(p_{i+1},r')$ and $P_{S_1\cup  \{r, p_{i+1}\}}(r,p_{i+1})$ with at most $n_2-\rr(S_2)$ and $n_1+1-\rr(S_1\cup\{r,p_{i+1}\})$ crossings, respectively. See Figure~\ref{fig:subcase31}. The connection of both paths gives rise to a path on $S$ with endpoints $r$ and $ r'$ that, in addition, is $1$-plane since $\ch{(S_2)}$ and $\ch{(S_1\cup  \{r, p_{i+1}\})}$ do not intersect. Moreover, the number of crossings in that path is at most $$n+1-\br(S_2)-\br(S_1\cup \{r, p_{i+1}\})$$ as $\rr(S_2)=\br(S_2)$, $\rr(S_1\cup \{r, p_{i+1}\})=\br(S_1\cup \{r, p_{i+1}\})$ and $n_1+n_2=n$.
Therefore, it remains to prove that $\br(S_2)+\br(S_1\cup \{r, p_{i+1}\})\geq \br(S)+1$.

If point $u$ is red and $p_{i+1}\neq v$, then $p_{i+1}$ is a blue run of $S_1\cup\{r, p_{i+1}\}$ but not of $S$. Further, if $p_{i+1}=v$ then it is counted twice as a blue run, one in $S_1\cup\{r,p_{i+1}\}$ and the other in $S_2$. The situation is analogous if $v$ is red. If $u$ and $v$ are both blue, then the blue run of $S$ containing them is counted twice. Hence, $\br(S_2)+\br(S_1\cup\{r, p_{i+1}\})\geq \br(S)+1$.

\vspace{0.3cm}
\noindent{\emph{Case 3.2: the configuration $(S_2, p_{i+1}, r')$ is special.}}
\vspace{0.3cm}

Since all points of $C_r$ are red, there exist two consecutive red points $r'_1,r'_2\in C_r$ such that $r'_1\in S_1$ and $r'_2\in S_2$. Let $C_1=\{p_i=q_1,q_2,\ldots,q_k= r'_1\}$ be the sequence of consecutive points on the boundary of $\ch(S_1)$ ordered counterclockwise from $p_i=q_1$ to $r'_1=q_k$ (see Figure~\ref{fig:subcase32}). By construction, all these points are visible from $r$. Let $q_{j+1}$ be the first red point of $C_1$, which exists since $p_i=q_1$ is blue and $r'_1=q_k$ is red. Analogously, let  $C_2=\{p_{i+1}=q'_1, q'_2, \ldots , q'_{k'} = r'_2\}$ be the sequence of consecutive points on the boundary of $\ch(S_2)$ ordered clockwise from $p_{i+1}=q'_1$ to $r'_2=q'_{k'}$, and let $q'_{j'+1}$ be the first red point of $C_2$.

\begin{figure}[!htbp]
\centering
\includegraphics[scale=0.7]{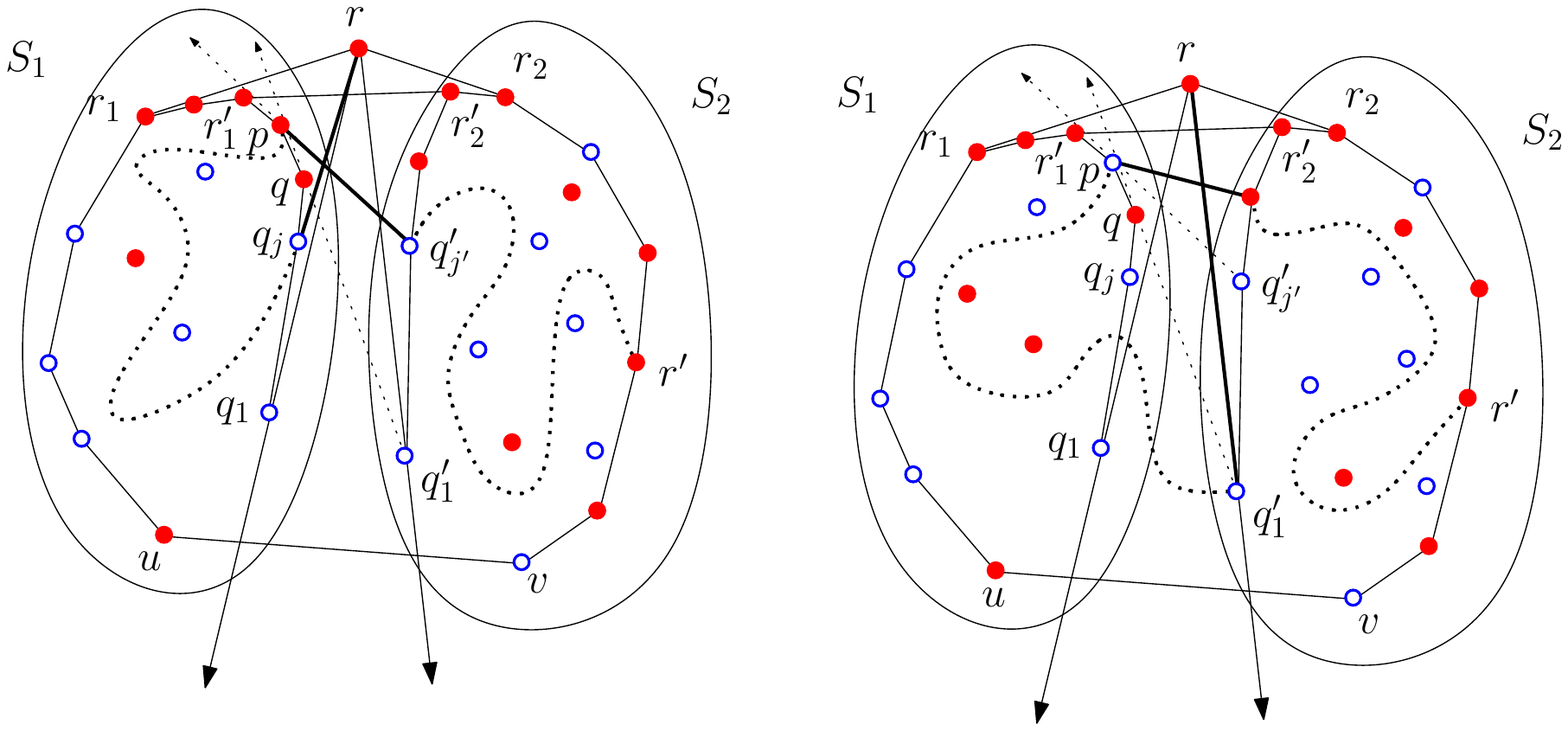}
\caption{$(S_2,p_{i+1},r') $ is a special configuration. Left: point $p$ is red. Right: point $p$ is blue.}\label{fig:subcase32}
\end{figure}

By construction, one of the tangent points from $q'_1$ to $S_1$, say $q$, belongs to $C_1$ and, for any other point of $C_2$, its corresponding tangent point is also $q$ or a point of $C_1$ placed after $q$ counterclockwise. Let $p$ be the tangent point of $C_1$ from $q'_{j'}$ to $S_1$.

Suppose first that point $p$ is red (see Figure~\ref{fig:subcase32} left). By induction, there are 1-plane paths $P_{S_1}(q_j,p)$ and $P_{S_2}(q'_{j'},r')$ with at most $n_1-\br(S_1)$ and $n_2-\br(S_2)$ crossings, respectively; note that the corresponding configurations are non-special since points $q_j$ and $q'_{j'}$ are blue with red neighbors $q_{j+1}$ and $q'_{j'+1}$. Therefore, $$P_S(r,r')=\{rq_j\}\cup P_{S_1}(q_j,p)\cup \{pq'_{j'}\}\cup P_{S_2}(q'_{j'}, r')$$ is a $1$-PHAP on $S$ with endpoints $r,r'$ and at most $n_1-\br(S_1)+n_2-\br(S_2)+1$ crossings. Observe that $\ch(S_1)$ and $\ch(S_2)$ do not intersect, and by construction, edges  $pq'_{j'}$ and $rq_j$ only produce one crossing that is their intersection point located outside of those convex hulls, and so one can guarantee the 1-plane character of $P_S(r,r')$. That $n-\br(S)$ is an upper bound for the number of crossings comes from the fact that $\br(S_2)+\br(S_1)\geq \br(S)+1$ (analogous to case 3.1) and $n=n_1+n_2$.

Assume now that point $p$ is blue (see Figure~\ref{fig:subcase32} right). Point $q'_{j'+1}\in \ch(S_2)$ is red and has a blue neighbor $q'_{j'}$, then $q'_{j'+1}\neq r'$ since the configuration $(S_2, p_{i+1}, r')$ is special. Hence, by induction, there are 1-plane paths $P_{S_2\setminus \{q'_1\}}(q'_{j'+1},r')$ and $P_{S_1\cup \{q'_1\}}(q'_{1},p)$ with at most $n_2-1-\br(S_2\setminus \{q'_1\})$ and $n_1-\rr(S_1\cup \{q'_1\})$ crossings, respectively. A $1$-plane path $P_S(r,r')$ is obtained by connecting the two paths via the edges $rq'_1$ and $pq'_{j'+1}$, that is
$$P_S(r,r') = \{rq'_1\}\cup P_{S_1\cup \{q'_1\}}(q'_{1},p)\cup  \{pq'_{j'+1}\}\cup P_{S_2\setminus \{q'_1\}}(q'_{j'+1},r').$$
Again, the $1$-plane character is guaranteed since $\ch(S_1\cup \{q'_1\})$ and $\ch(S_2\setminus \{q'_1\})$ do not intersect, and edges $rq'_1$ and $pq'_{j'+1}$ only produce one crossing (their intersection point) that is located outside the convex hulls as the endpoints of the edges are visible.

The number of crossings in $P_S(r,r')$ is at most $$n_2-1-\br(S_2\setminus \{q'_1\})+n_1-\rr(S_1\cup \{q'_1\})+1$$ which is $n-\br(S_2\setminus \{q'_1\})-\br(S_1\cup \{q'_1\})$; note that $\ch(S_1\cup \{q'_1\})$ is bichromatic, and so $\rr(S_1\cup \{q'_1\})=\br(S_1\cup \{q'_1\})$. Thus, it remains to prove that $\br(S_2\setminus \{q'_1\})+\br(S_1\cup \{q'_1\})\geq \br (S)$.

Since $(S_2, r', p_{i+1})$ is special, the two neighbors of $p_{i+1}=q'_1$ on the boundary of $\ch(S_2)$ are blue. Hence, either $q'_1$ produces a blue run in $S_1\cup \{q'_1\}$  (if $u$ is red), or the neighbors of $q'_1$ produce a blue run in $S_2\setminus \{q'_1\}$ (if $v$ is red), or the blue run of $S$ containing points $u$ and $v$ is counted twice (if $u$ and $v$ both are blue). In any case, $\br(S_2\setminus \{q'_1\})+\br(S_1\cup \{q'_1\})> \br (S)$.

\vspace{0.3cm}
\noindent{\emph{Case 4: Both $C_r$ and $C_{r'}$ consist of blue points.}}
\vspace{0.3cm}

Let $b\in C_r$. The configuration $(S\setminus \{r\},b,r')$ is non-special ($C_{r'}$ consists of blue points and so $r'$ has a blue neighbor on the boundary of $\ch(S\setminus \{r\})$) then, by induction, there is a 1-plane path $P_{S\setminus \{r\}}(b,r')$ with at most $n-\br(S\setminus \{r\})$ crossings. One might extend this path by the edge $rb$  but this would give $n-\br(S)+1$ as upper bound for the number of crossings. Indeed, the two neighbors of $r$ on the boundary of $\ch(S)$ are blue points, say $b_1,b_2$ (clockwise), and so the two blue runs of $S$ containing, respectively, points $b_1$ and $b_2$ produce a unique blue run in $S\setminus \{r\}$. Thus $\br(S\setminus \{r\})=\br(S)-1$. To reach the upper bound $n-\br(S)$, we provide a more elaborated construction by distinguishing two cases.

\vspace{0.3cm}
\noindent{\emph{Case 4.1: A red point $r''\in \ch(S\setminus \{r,b_2\})$ is visible from $b_2$.}}
\vspace{0.3cm}

\begin{figure}[!htbp]
\centering
\includegraphics[scale=0.7]{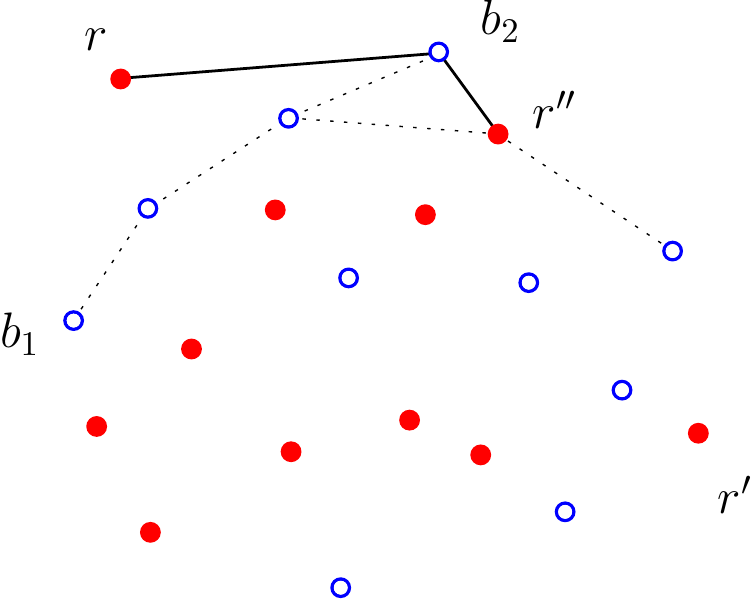}
\caption{A red point $r''\in \ch(S\setminus \{r,b_2\})$ is visible from $b_2$.}\label{fig:subcase41}
\end{figure}

By induction, there is a 1-plane path $P_{S\setminus \{r,b_2\}}(r'',r')$  with at most $n-1-\br(S\setminus \{r,b_2\})$ crossings. Thus, $P_S(r,r')=\{rb_2\}\cup \{b_2r''\}\cup P_{S\setminus \{r,b_2\}}(r'',r') $ is
clearly a $1$-PHAP on $S$ that also has at most $n-1-\br(S\setminus \{r,b_2\})$ crossings. See Figure~\ref{fig:subcase41}.

Since $r''$ is on the boundary of $\ch(S\setminus \{r,b_2\})$ then $\br(S\setminus \{r,b_2\})\geq \br(S\setminus \{r\})=\br(S)-1$. Therefore, the number of crossings in $P_S(r,r')$ is at most $n-1-(\br(S)-1) = n-\br(S)$.

\vspace{0.3cm}
\noindent{\emph{Case 4.2: Every point of $\ch(S\setminus \{r,b_2\})$ visible from $b_2$ is blue.}}
\vspace{0.3cm}

By Lemma~\ref{lemm:partition}$(i)$ on $S\setminus \{r\}$, there is a counterclockwise partition $S_1\cup S_2$ of $S\setminus \{r, b_2\}$ around $b_2$, where $p_i$ and $p_{i+1}$ are red. Suppose first that $r'\neq p_{i+1}$, and  assume that $r'\in S_1$ (a symmetric construction can be done for $r' \in S_2\setminus \{p_{i+1}\}$); see Figure~\ref{fig:subcase42} left.

\begin{figure}[!htbp]
\centering
\includegraphics[scale=0.7]{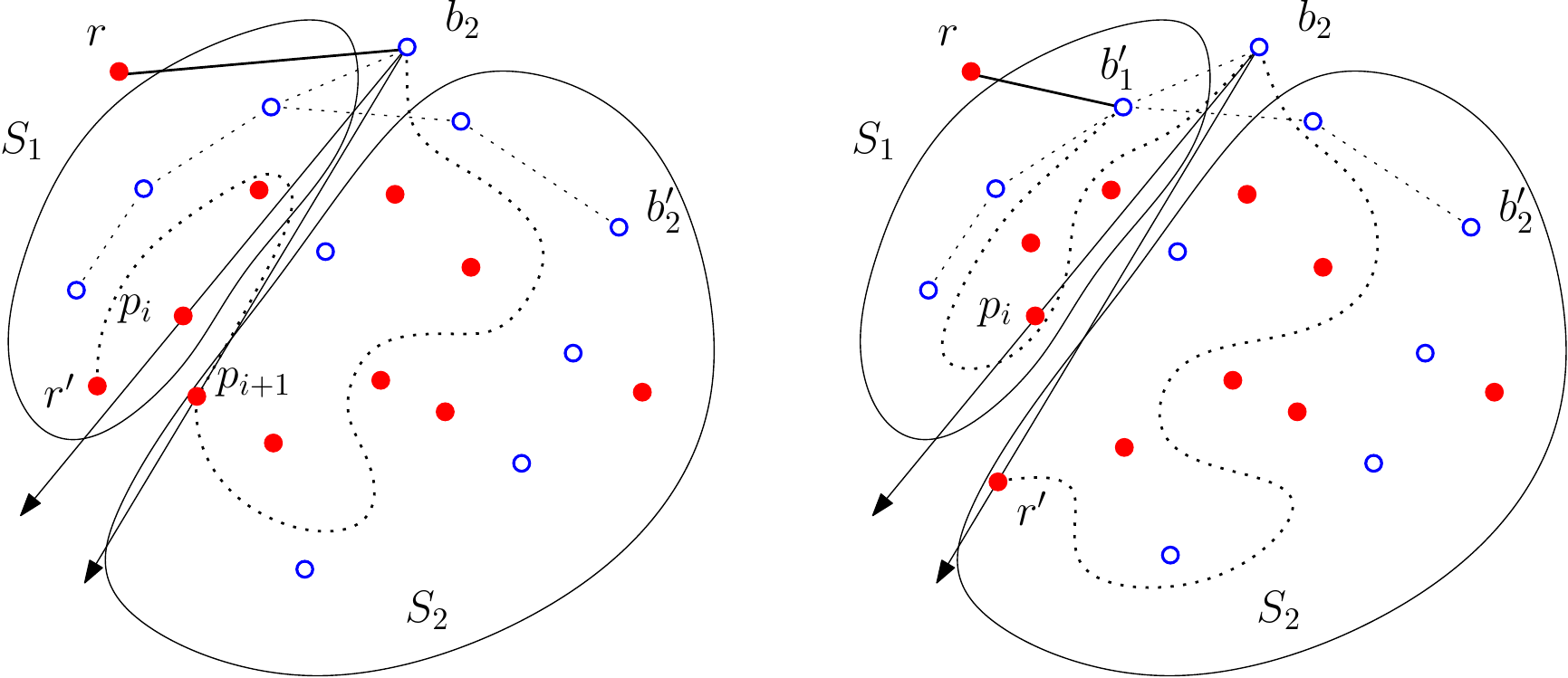}
\caption{Every point of $\ch(S\setminus \{r,b_2\})$ visible from $b_2$ is blue.}\label{fig:subcase42}
\end{figure}

Points $b_2$ and $p_{i+1}$ are neighbors on $\ch(S_2\cup \{b_2\})$ and have distinct color, then the configuration $(S_2\cup \{b_2\}, p_{i+1}, b_2)$ is non-special and, by induction,  there is a 1-plane path $P_{S_2\cup \{b_2\}}(b_2, p_{i+1})$ with at most $n_2+1-\br(S_2\cup \{b_2\})$ crossings, where $n_2=|(B\setminus\{b_2\})\cap S_2|$.
Again by induction, there is a 1-plane path $P_{S_1\cup \{p_{i+1}\}} (p_{i+1},r')$ with at most $n_1-\br(S_1\cup \{p_{i+1}\})$ crossings, where $n_1=|(B\setminus\{b_2\})\cap S_1|$. Thus, $$P_S(r,r') = \{rb_2\} \cup P_{S_2\cup \{b_2\}}(b_2,p_{i+1})\cup P_{S_1\cup p_{i+1}} (p_{i+1},r')$$ is a $1$-PHAP with endpoints $r$ and $r'$ and at most $n-\br(S_2\cup \{b_2\})-\br(S_1\cup \{p_{i+1}\})$ crossings (note that $n=n_1+n_2+1$). The same arguments as those of case 3 prove that $ P_S(r,r')$ is 1-plane. With respect to the number of crossings, the bound of $n-\br(S)$ is obtained from the fact that $$\br(S_2\cup \{b_2\})+\br(S_1\cup \{p_{i+1}\})\geq \br((S\setminus \{r\})+1=\br(S),$$ note that the blue run of $S\setminus \{r\}$ containing point $b_2$ is counted twice, one in $S_2\cup \{b_2\}$  and the other in $S_1\cup \{p_{i+1}\}$.

\begin{figure}[!htbp]
\centering
\includegraphics[scale=0.7]{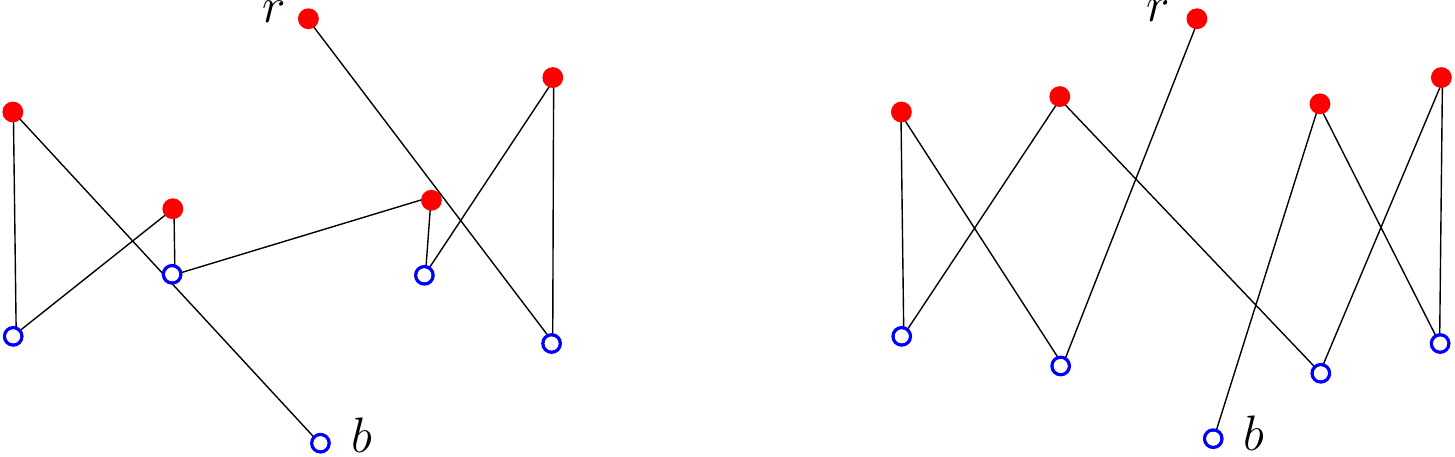}
\caption{A special configuration $(S,r,b)$ that: (left) admits a $1$-PHAP, (right) does not admit a $1$-PHAP.}\label{fig:configurations}
\end{figure}

Suppose now that $r'=p_{i+1}$ (see Figure~\ref{fig:subcase42} right). Let $b'_1$ be the previous point (clockwise) to $b_2$ on the boundary of $\ch(S\setminus \{r\})$ which, by construction, is blue and visible from $r$. The desired $1$-PHAP on $S$ is $$P_S(r,r') = \{rb'_1\}\cup P_{S_1\cup b_2}(b'_1,b_2)\cup P_{S_2\cup b_2}(b_2, r')$$ where $P_{S_1\cup b_2}(b'_1,b_2)$ and $P_{S_2\cup \{b_2\}}(b_2, r')$ are $1$-PHAP obtained by induction on $S_1\cup \{b_2\}$ and $S_2\cup \{b_2\}$, respectively; note that  $r'$ and $b_2$ are neighbors of distinct color on $\ch(S_2\cup \{b_2\})$ which implies that the configuration $(S_2\cup \{b_2\},b_2, r')$ is non-special.
The number of crossings in $P_S(r,r')$ is at most $n-\br(S)$ since again the blue run of $S\setminus \{r\}$ containing $b_2$ is counted twice.
\end{proof}

\begin{remark}
One cannot guarantee the existence of a $1$-PHAP for a special configuration. For instance, Figure~\ref{fig:configurations} shows two similar special configurations: the one on the left admits a $1$-PHAP and the other (right) does not. In this last configuration (which is also shown in Figure~\ref{fig:cycle} as an example of special configuration), at least one of the edges must be crossed twice in any Hamiltonian alternating path because the points are essentially in convex position, and for point sets in convex position, the only configurations that do not admit a $1$-PHAP are precisely the special configurations, as we prove in the next section (see Theorem~\ref{the:specialconfig}).
\end{remark}

We are now ready to prove our main result in this section.

\begin{theorem}\label{the:cycle}
For given $S=R\cup B$ with $|B|=|R|=n\geq 2$, there exists a 1-plane Hamiltonian alternating cycle on $S$ with at most $n-\max \{\rr(S), \br(S)\}$ crossings.
\end{theorem}

\begin{proof}
Suppose first that  $\ch(S)$ is bichromatic (which implies $\rr(S)=\br(S)$) and consider two consecutive points of distinct color, say a red point $r$ and a blue point $b$. Thus, the configuration $(S,r,b)$ is  non-special and, by Lemma~\ref{lemm:path}, there is a $1$-PHAP on $S$ with endpoints $r$ and $b$ and at most $n-\rr(S)$ crossings. By adding the edge $rb$,  the desired cycle is obtained.

Assume now that the boundary of $\ch(S)$ consists of points of only one color, say red ($\rr(S)=1$ and $\br(S)=0$), and let $r$ be one of those points. Partition counterclockwise $S\setminus \{r\}$ around $r$ as described in Lemma~\ref{lemm:partition}$(i)$ obtaining the sets $S_1, S_2$ and the point $p_{i+1}$.

By Lemma~\ref{lemm:path}, there are 1-plane paths $P_{S_1\cup \{r, p_{i+1}\}}(r,p_{i+1})$ and  $P_{S_2\cup \{r\}}(p_{i+1},r)$ with at most $n_1+1-\rr(S_1\cup \{r, p_{i+1}\})$ and $n_2+1-\rr(S_2\cup \{r\})$ crossings, respectively. Connecting both paths, we construct a $1$-PHAC on $S$ with at most $$n_1+n_2+2-\rr(S_1\cup \{r, p_{i+1}\})-\rr(S_2\cup \{r\})$$ crossings. Since $n_1+n_2+1=n$ and, in the worst case $\rr(S_1\cup \{r, p_{i+1}\})=\rr(S_2\cup \{r\})=1$, the cycle has at most $n-1$ crossings.
\end{proof}

\begin{remark}
Theorem \ref{the:cycle} can also be proved by using the Ham-sandwich theorem and Lemma~\ref{lemm:path}. Further, the bound of $n-\max \{\rr(S), \br(S)\}$ for the number of crossings is tight. Indeed, Kaneko et al.~\cite{KKY00} proved that, for a point set $S$ consisting of $2n$ points in convex position, $n$ consecutive red points and $n$ consecutive blue points, any Hamiltonian alternating cycle has at least $n-1$ crossings. However, there are configurations of points for which this bound is far to be tight. For instance, if the $n$ red points are on the boundary of $\ch(S)$ and the $n$ blue points are inside $\ch(S)$, then it is not difficult to prove that there always exists a Hamiltonian alternating cycle on $S$ with no crossings (see Lemma 3.4 of~\cite{GT95}), and the bound of Theorem  \ref{the:cycle} would be $n-1$.
\end{remark}

To conclude this section, we show that Lemma~\ref{lemm:path} and Theorem~\ref{the:cycle} let us compute a $1$-PHAC on $S=R\cup B$ (not necessarily minimizing the number of crossings) in $O(n^2)$ time and space where, as usual,  $n=|R|= |B|$.

An $O(n^2\log n)$ time and $O(n)$ space algorithm for this problem is clear, as Kaneko et al. pointed out in~\cite{KKY00}: each time that one (or two) edges are assigned to the cycle, we must compute a convex hull of one (or two) subset, and the radial order of the points of a given subset around a given point $r$. In principle, computing a convex hull and a radial order can be done in $O(n\log n)$ time and, since this is done $O(n)$ times, the overall complexity is $O(n^2\log n)$ time and $O(n)$ space. However, these two operations can be done more efficiently as follows.

All the radial orders around the $2n$ points of  $S$ can be obtained in $O(n^2)$ time and space (see for instance~\cite{OW88}). Thus, in a generic step of the above algorithm, computing the radial order of the points of a subset $S'$ around a point $r$ requires only linear time.

Once $\ch(S')$ has been computed, its updates after point deletions can be done in $O(\log n)$ amortized time per deletion, using the semi-dynamic data structure of Hershberger and Suri~\cite{HS92}. Therefore, computing the different convex hulls in the overall process only requires $O(n \log n)$. Moreover, the same semi-dynamic data structure allows to compute a tangent from a point outside $\ch(S')$ with the same complexity, $O(\log n)$ amortized time. Thus, we obtain the following corollary.

\begin{corollary}\label{cor:cyclegp}
For given $S=R\cup B$ with $|B|= |R|=n \geq 2$, a 1-plane Hamiltonian alternating cycle on $S$ can be computed in $O(n^2)$ time and space.
\end{corollary}

\section{$1$-PHAC and $1$-PHAP for point sets in convex position}\label{sec:convex}

In this section, we study 1-plane hamiltonian alternating cycles and paths for a well-known restricted position of the point set: the convex position. Thus, hereafter the set $S=R\cup B$ is assumed to be in convex position, unless otherwise stated.

\subsection{1-plane Hamiltonian alternating cycles}\label{subsec:cycles}

We first prove that every Hamiltonian alternating cycle on $S$ has at least $n-\max \{\rr(S), \br(S)\}=n-\rr(S)$ crossings (note that $\rr(S)=\br(S)$), and moreover, only those cycles that are 1-plane attain the bound. We shall write \emph{optimum} for a path or cycle with  minimum number of crossings. The two endpoints of a run of $S$ are called \emph{limits}, and a \emph{bridge} of $S$ is a bichromatic edge that connects consecutive runs of $S$, i.e., its endpoints are the two consecutive points of $S$ that are limits of those runs.

\begin{theorem}\label{the:convexcycles}
For given $S=R\cup B$ in convex position with $|R|=|B|=n\geq 2$, the following statements hold.
\begin{itemize}
  \item[(i)] Every Hamiltonian alternating cycle on $S$ has at least $n-\rr(S)$ crossings.
  \item[(ii)] A Hamiltonian alternating cycle on $S$ with $n-\rr(S)$ crossings is 1-plane.
  \item[(iii)] Every bridge of $S$ belongs to any $1$-PHAC on $S$ with $n-\rr(S)$ crossings.
\end{itemize}
\end{theorem}

\begin{proof}
Let $C$ be a Hamiltonian alternating cycle on $S$. Given a red run of $S$ of cardinality $k$, there are at most two edges of $C$ that are bridges of $S$ and have an endpoint in the run. Thus, the remaining edges of $C$ with an endpoint in that red run, at least $2k-2$, are necessarily chords. Extending this argument to every red run of $S$, we obtain that at least $2n-2\rr(S)$ edges of $C$ are chords. Further, every chord must be crossed at least once to connect a point to the left to a point to the right of the chord. Since a crossing requires two chords and a chord must be crossed at least once, then the number of crossings in $C$ is at least $n-\rr(S)$; this proves statement (i).
Moreover, if the number of crossings is precisely $n-\rr(S)$, then every chord is crossed exactly once and every bridge of $S$ is an edge of $C$. Thus, statements $(ii)$ and $(iii)$ follow.
\end{proof}

\begin{remark}\label{rem:converse}
The converse of statement $(ii)$ in the preceding theorem is not true: it is possible to construct a $1$-PHAC on a set $S$ with more than $n-\rr(S)$ crossings. For instance, if $S$ consists of points alternating in color then $\rr(S)=n$, and obviously one can draw a $1$-PHAC without crossings, but also a $1$-PHAC with $\lceil \frac{n}{2} \rceil$ or $\lceil \frac{n}{2} \rceil+1$ crossings (depending on the parity of $n$). Moreover, the number of different 1-plane Hamiltonian alternating cycles on a point set can be exponential. Indeed, if $S$ consists of points alternating in color, for every subset $P=\{p_j, p_{j+1}, p_{j+2}, p_{j+3}\}$ of four consecutive points we can connect the points either in the same order $p_j, p_{j+1}, p_{j+2}, p_{j+3}$ or in the order $p_j, p_{j+3}, p_{j+2}, p_{j+1}$. Now,  divide $S$ into $\frac{n}{2}$ groups of four consecutive points each, and connect each group to the next one, clockwise, to form a cycle: either using  edge $p_{j+3}p_{j+4}$  or  edge $p_{j+1}p_{j+4}$. Clearly, the resulting cycle is 1-plane; the only possible crossings correspond to edges of type $p_jp_{j+3}$ and $p_{j+1}p_{j+4}$. Since for every group there are two options to connect its points, then there are $\Omega (2^{n\over 2})$ different 1-plane Hamiltonian alternating cycles on $S$.
\end{remark}

Observe that, by removing a bridge of an optimum $1$-PHAC, we obtain a $1$-PHAP whose endpoints are two consecutive points of $S$ of different colors. Hence, an optimum $1$-PHAP connecting two consecutive points of different colors inherits the properties of an optimum $1$-PHAC.  When $|R|=|B|+1=n+1$, a similar situation happens for an optimum $1$-PHAP on $S$ with two consecutive red points $r$ and $r'$ as endpoints; one can add a dummy blue point between $r$ and $r'$ so that the path can be completed to a cycle. Thus, we have the following corollary.

\begin{corollary}\label{cor:consecutive}
For given $S=R\cup B$ in convex position, the following statements hold.
\begin{itemize}
  \item[(i)] Let $r\in R$ and $b\in B$ be consecutive points on $S$. If $|R|=|B|=n\geq 2$, then every optimum Hamiltonian alternating path on $S$ with endpoints $r$ and $b$ has $n-\rr(S)$ crossings, is 1-plane and uses all  bridges of $S$ except for $rb$.
  \item[(ii)] Let $|R|=|B|+1=n+1\geq 2$ and let $r, r'\in R$ be consecutive red points on $S$. Then every optimum  Hamiltonian alternating path on $S$ with endpoints $r$ and $r'$ has $n-\br(S)$ crossings, is 1-plane and uses all  bridges of $S$.
\end{itemize}
\end{corollary}

We now prove that an optimum $1$-PHAC on a set $S$ in convex position can be computed in $O(n)$ time and space. Recall that, by Corollary \ref{cor:cyclegp}, to find a $1$-PHAC for a set in general position takes $O(n^2)$ time and space, and the cycle in that case is not necessarily optimum.

We begin by describing a linear procedure to compute, for every red point, the first clockwise blue point such that, between them, there are the same number of red and blue points. For simplicity in the arguments, in the explanation below, we modify the notation used in the proof of Lemma \ref{lemm:path} for the endpoints of the paths: they were called $r,b,r'$, etc, and now they will be denoted as $p_i,p_j$, etc. We shall go back to that original notation to state Theorem \ref{Theorem}.

Consider the clockwise circular order $<p_1,\ldots, p_{2n}>$ of the points of $S$, where $|R|=|B|$. Let $S[p_i,p_j]=\{ p_i,p_{i+1}\ldots,p_{j}\}$ (clockwise) and $S(p_i,p_j)=S[p_i,p_j]\setminus\{p_i,p_j\}$. For each red point $p_i$, let $p_{J(i)}$ be the first point of $S$ such that $|S[p_i,p_{J(i)}]\cap R|=|S[p_i,p_{J(i)}]\cap B|$.

By Lemma \ref{lemm:partition} (in its clockwise version), if $p_{i+1}$ is blue then $p_{J(i)}=p_{i+1}$, otherwise $p_{J(i)-1}$ and $p_{J(i)}$ are both blue. Note that $J(i)\neq J(i')$ whenever $i\neq i'$, and so the mapping that assigns to each value of $i$ the value $J(i)$ is bijective.

The following procedure computes point $p_{J(i)}$ for each red point $p_i$.

\vspace{0.2cm}
\textsc{Procedure: $J$-PAIRS}

\textbf{Input:} $S=B\cup R$ in convex position, $|B|=|R|=n$

\textbf{Output:} All pairs $(p_i, p_{J(i)})$ for $p_i\in R$

\begin{enumerate}
\item Consider two sorted lists of red and blue points: $L_R$ is the red list and $L_B$ is the blue list. Initially both lists are empty.

\item Explore the point set $S=<p_1,\ldots, p_{2n}>$ given in clockwise circular order:
\begin{enumerate}
\item If the explored point $p_j$ is red, add it at the end of list $L_R$.
\item If the explored point $p_j$ is blue and $L_R\neq \emptyset $, assign $p_j$ to the last red point of list $L_R$ and remove that red point from $L_R$.
\item If the explored point $p_j$ is blue and $L_R=\emptyset$, add $p_j$ at the end of list $L_B$.
 \end{enumerate}
\item  When step 2 has finished, $L_R$ and $L_B$ have the same number of points. Assign the remaining points of $L_B$ to the remaining points of $L_R$ in reverse order, that is, the first point of $L_B$ is assigned to the last point of $L_R$. Remove both points from the lists, and repeat the assignation last-first until both lists are empty.
\end{enumerate}
\textsc{End Procedure}

\vspace{0.2cm}

Figure~\ref{AlgorithmOptimalCycle} illustrates the assignation of points described in the above procedure. The assignations done in step 3 (with the points that remain in lists $L_R$ and $L_B$) are represented by a thicker line than those of step 2.

\begin{figure}[tb]
\centering
\includegraphics[scale=0.7]{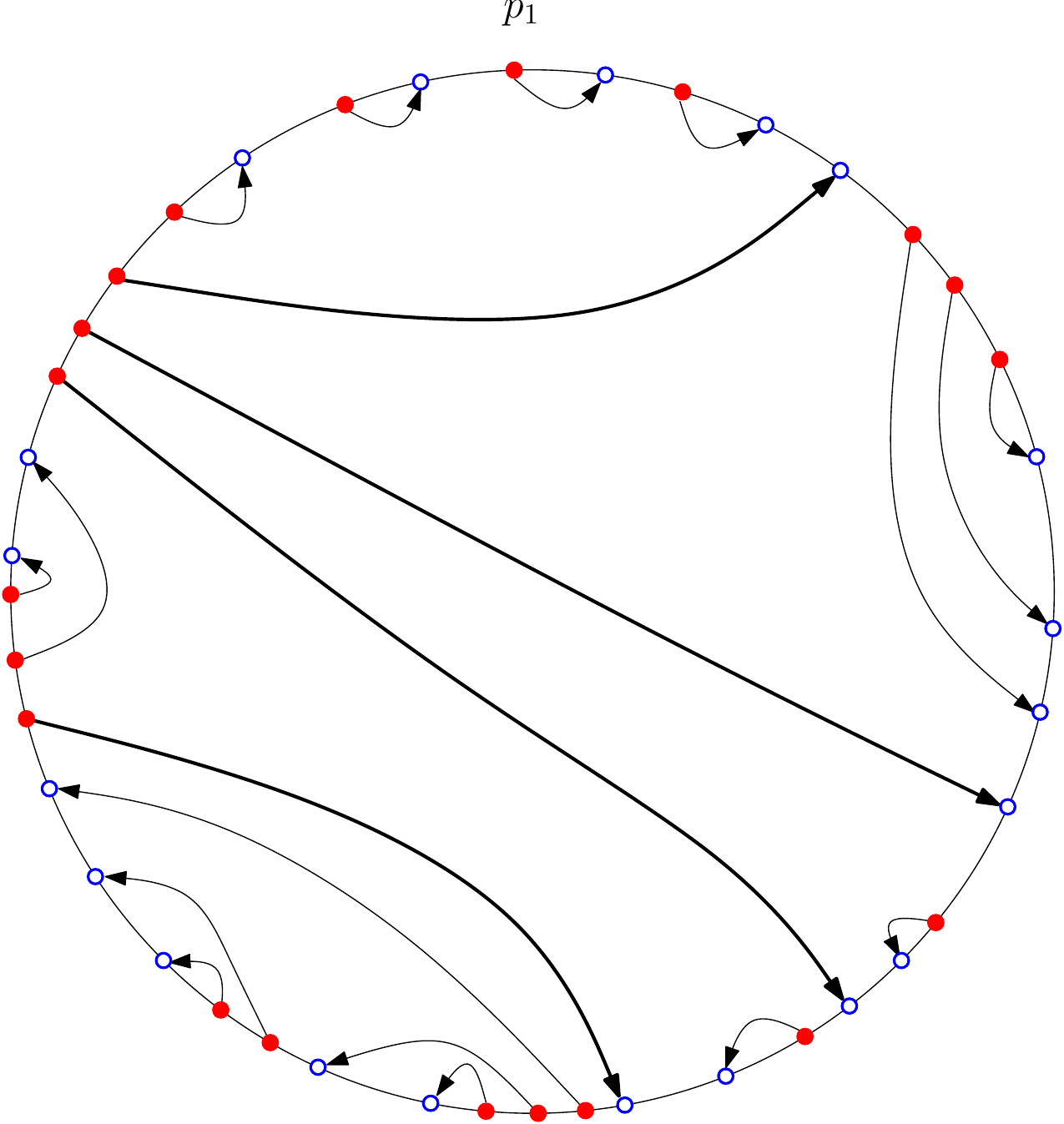}
\caption{Computing pairs $(p_i, p_{J(i)})$.}\label{AlgorithmOptimalCycle}
\end{figure}

\begin{lemma}\label{Auxiliar}
\textsc{Procedure $J$-PAIRS} computes all pairs $(p_i,p_{J(i)})$ with $p_i\in R$ in $\Theta(n)$ time.
\end{lemma}

\begin{proof}
Observe first that, in step 2, a point in list $L_B$ is never removed from that list, and it can be included in $L_B$ only if $L_R=\emptyset$; a point in list $L_R$ can be removed from the list, but in that case it is never inserted again.

We next prove that if $p_j$ is assigned to $p_i$ in \textsc{Procedure $J$-PAIRS} then $p_j=p_{J(i)}$, that is, $|S[p_i,p_j]\cap B|=|S[p_i,p_j]\cap R|$ and $p_j$ is the first point after $p_i$, clockwise, satisfying that equality.

Suppose that such result is true for all the assignments done in step 2 before exploring point $p_j$, and then $p_j$ must be assigned to $p_i$. This means that $p_j$ is blue and $p_i$ is the last red point in list $L_R$. Further, all the blue points in $S(p_i,p_j)$ have been assigned to red points in the same set $S(p_i,p_j)$ (otherwise $p_i$ cannot be in list $L_R$), and all the red points in $S(p_i,p_j)$ have received an assignment of a blue point in $S(p_i,p_j)$ (otherwise they should be in $L_R$ after $p_i$). Therefore, $|S[p_i,p_j]\cap B|=|S[p_i,p_j]\cap R|$.

Now, if $j$ is not the first index such that $J(i)=j$, but there is an index $j'<j$ so that $J(i)=j'$, then $j'$ would have been previously assigned, i.e.,  $j'=J(i')$ for some other point $p_{i'}$. This contradicts that  $J$ is bijective. Thus, if $p_j$ is assigned to $p_i$ in step 2 then $p_j=p_{J(i)}$.

In step 3, before the last point $p_k\in L_R$ and the first point $p_{\ell}\in L_B$ (clockwise), all the red and blue points have been assigned. Reasoning as above, $|S[p_k,p_{\ell}]\cap B|=|S[p_k,p_{\ell}]\cap R|$ and  $p_{\ell}$ is the first point after $p_k$, clockwise, satisfying that equality. Thus, $p_{\ell}=p_{J(k)}$. If we now remove these two points from the corresponding lists, the new last point of $L_R$ and the new first point  of $L_B$ are again in the same conditions, and so the same argument can be applied.
\end{proof}

If \textsc{Procedure $J$-PAIRS} is performed counterclockwise, we obtain another blue point $p_{J'(i)}$ for each red point $p_i$. When $|R|= |B|+1$, the blue point $p_{J(i)}$ does not exist for one of the red points, but the same procedure computes $p_{J(i)}$ for the remaining red points. Note that, if $p_k$ is the red point for which $p_{J(k)}$ is not defined, then $p_{J'(k)}$ necessarily exists.

The special configuration of Definition \ref{def:special} (for $S$ in convex position) occurs when $p_{J(i)}=p_{J'(i)}$; points $r$ and $ b$ in that definition would be $p_i\in R$ and $p_{J(i)}\in B$. Moreover, if the assignation of \textsc{Procedure $J$-PAIRS} is done counterclockwise and blue points play the role of red points, we obtain the same assignation as above: the red point $p_i$ is assigned to the blue point $p_{J(i)}$.

With \textsc{Procedure $J$-PAIRS} in hand, we can compute optimum Hamiltonian alternating cycles or paths on $S$ in $O(n)$ time and space.

\begin{theorem}\label{Theorem}
For given $S=R\cup B$ in convex position, the following statements hold.
 \begin{itemize}
  \item[(i)] An optimum Hamiltonian alternating cycle on $S$ can be computed in $O(n)$ time and space provided that $|R|= |B|=n \geq 2$.
 \item[(ii)] Let $r\in R$ and $b\in B$ consecutive points on $S$. If $|R|=|B|=n\geq 2$, then an optimum Hamiltonian alternating path on $S$ with endpoints $r$ and $b$ can be computed in $O(n)$ time and space.
 \item[(iii)] Let $|R|=|B|+1=n+1\geq 2$, and let $r, r'\in R$ be consecutive red points on $S$. Then, an optimum Hamiltonian alternating path on $S$ with endpoints $r$ and $ r'$ can be computed in $O(n)$ time and space.
 \end{itemize}
Moreover, in all cases, the optimum is $1$-plane.
\end{theorem}

\begin{proof}
By Theorem~\ref{the:convexcycles} and Corollary~\ref{cor:consecutive} the optimum is 1-plane in all cases, and proving statements $(i)$ and $(ii)$ is equivalent. Hence, it suffices to show that the process described in Lemma~\ref{lemm:path} to compute $1$-plane Hamiltonian alternating paths can be done in linear time.

Suppose that $|R|=|B|=n$, and assume that points $p_{J(i)}$ have been pre-computed in linear time (regardless of the color of the points $p_i$) by applying \textsc{Procedure $J$-PAIRS} twice, one for the red points and the other for the blue points.

For the special case of points in convex position, in a generic step of the process described in the proof of Lemma~\ref{lemm:path}, we only have to compute optimum 1-plane Hamiltonian alternating paths with endpoints $p_i$ and $p_j$ on some subsets of consecutive points $S[p_i,p_j]$.

If $p_i$ and $p_j$ have different colors, in constant time, we can compute one or two edges of that optimum path, and study a smaller subproblem. The same happens when $p_i$ and $p_j$ have the same color, and one of them has a neighbor of different color.
Now, if $p_i$, $p_j$ and their neighbors have the same color, in constant time we can study two subproblems by using point $p_{J(i)}$ that, by Lemma~\ref{lemm:partition}, has a different color. Thus, in all cases, the cost of computing each new edge of the desired optimum path can be done in constant time. This implies that the procedure described in Lemma~\ref{lemm:path} is $O(n)$.

A similar argument applies when $|R|=|B|+1$ and the desired optimum $1$-PHAP on $S$ has two red consecutive points, $r,r'\in R$, as endpoints. In the pre-processing step, if $r'$ is after $r$ (clockwise), then it suffices to compute $p_{J(i)}$ beginning at $r'$ and $p_{J'(i)}$ beginning at $r$.
\end{proof}

\subsection{1-plane Hamiltonian alternating paths}\label{subsec:paths}

We first show that if $S=R\cup B$ is in convex position and $|R|=|B|$, the special configurations $(S,r,b)$ of Definition~\ref{def:special} are the unique configurations of points that do not admit a $1$-PHAP with fixed endpoints $r$ and $b$. Recall that by Lemma~\ref{lemm:path}, in any other case, such a path does exist, and also a $1$-PHAP with fixed red endpoints provided that $|R|=|B|+1$.

In this section, again, to simplify statements and explanations, the notation of Section~\ref{sec:general} for the endpoints of the paths will be modified when necessary:  $r,r',b, b'$ may be denoted as $p,q,p_s,p_t$, etc. We begin with a technical lemma, which considers that the $1$-PHAP is oriented: if $p_jp_k$  is an edge, point $p_j$ is visited before point $p_k$.

\begin{lemma}\label{lemm:1convex}
Let $S=R\cup B$ be in convex position with $|R|=|B|\geq 1$ or $|R|=|B|+1\geq 2$, and let $<p_1,\ldots, p_{h}>$ be the clockwise circular order of its points. Suppose that $P$ is an oriented $1$-PHAP on $S$ from $p=p_j$ to $q$, which contains a chord $p_jp_k$. If $q$ is to the left of the oriented chord $p_jp_k$, then $P$ visits the following points in the given order: (1) $p_k$, (2) all points of the set $\{p_{k+1}, \ldots , p_{j-2} \}$ (regarded cyclically), (3) $p_{j-1}$, (4) either $p_{j+1}$ or $p_{k-1}$.
\end{lemma}

\begin{proof}
The oriented chord $p_jp_k$ splits $S\setminus \{p_j,p_k\}$ into two subsets: $S_{\ell}$ containing all points that are to the left of the chord, and $S_r=\{p_{k+1}, p_{k+2}, \ldots , p_{j-1} \}$ whose points are to the right of the chord. Suppose that after $p_k$, path $P$ visits a point in $S_{\ell}$ before than a point in $S_r$. Then, to reach the endpoint $q\in S_{\ell}$ we must cross twice the chord $p_jp_k$, which contradicts that $P$ is 1-plane. Hence, $P$ visits all points of $S_r$  before crossing the chord $p_jp_k$ to go to a point in $S_{\ell}$. We now prove that the last point visited in $S_r$ is $p_{j-1}$.

Suppose on the contrary that $p'\in S_r\setminus \{p_{j-1}\}$ is the last visited point by $P$  in $S_r$. Then, there is a chord $p'q'$ with $q'\in S_{\ell}$ that crosses $p_jp_k$ and splits $S_{r}$ into two subsets: points to the left of the chord $p'q'$, and those that are to the right (here we include point $p_k$). Since $P$ connects both sets (which are non-empty) before crossing the chord $p_jp_k$, there must be another crossing in chord $p'q'$ (besides the intersection point with $p_jp_k$); a contradiction. Thus, $p'=p_{j-1}$. An analogous argument proves that $P$ visits either $p_{j+1}$ or $p_{k-1}$ just after $p_{j-1}$.
\end{proof}

\begin{theorem}\label{the:specialconfig}
Let $S=R\cup B$ be in convex position with $|R|=|B|=n\geq 1$. Given $r\in R$ and $b\in B$, there exists a 1-plane Hamiltonian alternating path on $S$ with endpoints $r$ and $ b$ if and only if the configuration $(S,r,b)$ is  non-special.
\end{theorem}

\begin{proof}
By Lemma~\ref{lemm:path}, if the configuration $(S,r,b)$ is  non-special then the desired path exists. We now prove
that there is no $1$-PHAP for a special configuration $(S,r,b)$. Suppose on the contrary that this path, say $P$, does exist, and consider it to be oriented from $r$ to $b$.

\begin{figure}[!htb]
\centering
\includegraphics[scale=0.7]{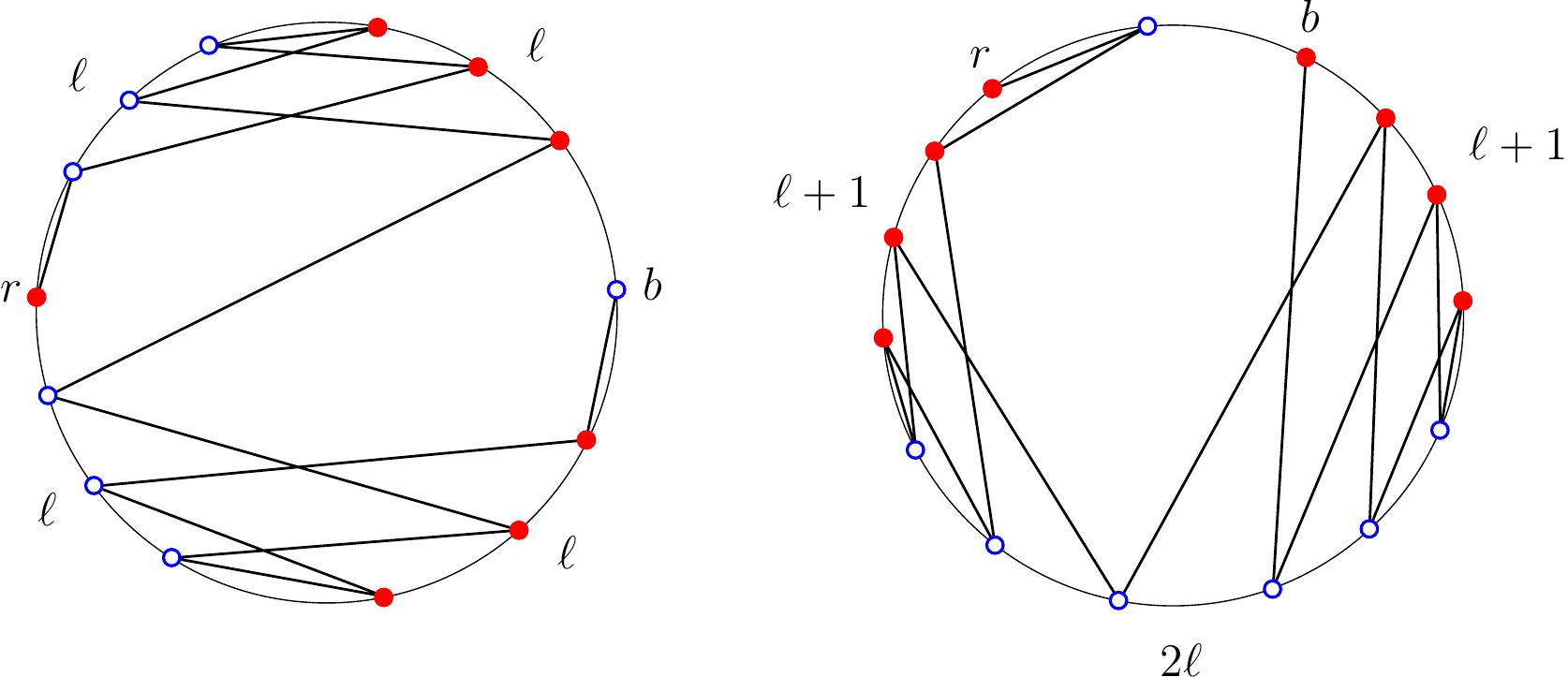}
\caption{Left: A point configuration that has two runs of cardinality 1 and four runs of cardinality $\ell$. It admits no $1$-PHAP with endpoints $r$ and $b$ and less than $2\ell-2$ crossings. Right: A point configuration that has one run of cardinality 1, two runs of cardinality $\ell+1$, and one run of cardinality $2\ell$. It admits no $1$-PHAP with endpoints $r$ and $b$ and less than $2\ell-2$ crossings.}\label{fig:tight}
\end{figure}

Following the notation of Lemma \ref{lemm:1convex}, $P$ visits a blue point $p_k$ just after $r=p_j$. Assume, without loss of generality, that $b=q$ is to the left of the oriented edge $p_jp_k$.

By Lemma~\ref{lemm:1convex}, and considering the clockwise circular order  $< p_1,\ldots, p_{2n}>$ of the points of $S$, path $P$ visits first all points of $\{p_{k+1},\ldots,p_{j-2} \}$, then point $p_{j-1}$ (which is red since $(S,r,b)$ is special), and then point $p_{k-1}$ that is blue (note that point $p_{j+1}$ is also red).

Since the path from $p_{k}$ to $p_{j-1}$ has endpoints of different color, the set of points from $p_{k}$ to $p_{j-1}$ (clockwise) contains the same number of red and blue points. Thus, points $p_k$ and $p_{k-1}$ define a partition of $S\setminus \{r\}$ around $p_j=r$, as described in Lemma~\ref{lemm:partition}. Further, $(S,r,b)$ is special and so $p_{k-1}=b$. Hence, path $P$ reaches the endpoint $b$ without visiting point $p_{k-2}$ which, by Definition \ref{def:special} of special configuration, exists and is blue. This contradicts that $P$ is Hamiltonian.
\end{proof}

\begin{remark}\label{remextra}
The bound of $n-\rr(S)$ given in Lemma~\ref{lemm:path} for the number of crossings in a $1$-PHAP is also tight when considering point sets $S$ in convex position. Figure~\ref{fig:tight} illustrates two point configurations in which every $1$-PHAP with endpoints $r$ and $b$ has at least $n-\rr(S)$ crossings. Thus, by Lemma~\ref{lemm:path}, at least one of those paths has exactly $n-\rr(S)$ crossings.
\end{remark}

We now turn to a natural problem for Hamiltonian alternating paths on non-special configurations: to determine whether being optimum implies being 1-plane.
This is the analogous problem  as that considered in Theorem  \ref{the:convexcycles}$(ii)$ for cycles, but there is a fundamental difference: the path will be assumed to have given endpoints.

We first state a well-known property based on the quadrangular property, and a consequence of it  (whose straightforward proof is omitted); see also Figure~\ref{fig:quadrangular}.

\begin{figure}[!htb]
\centering
\includegraphics[scale=0.7]{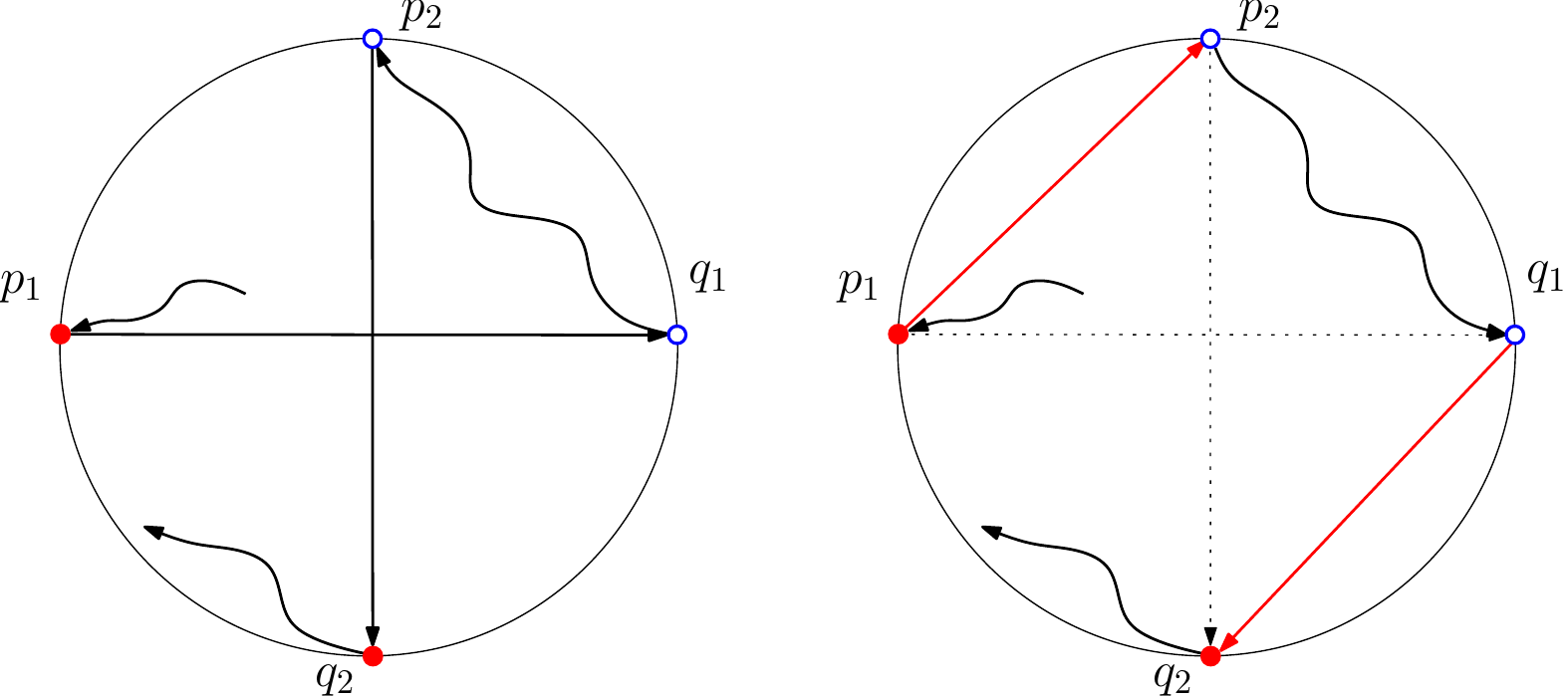}
\caption{Replacing edges $p_1q_1$ and $p_2q_2$ by edges $p_1p_2$ and $q_1q_2$.}\label{fig:quadrangular}
\end{figure}

\begin{property}\label{pro:quadrangular}
If $G$ is a geometric graph whose edges $pq$ and $p'q'$ intersect, then the geometric graph obtained by replacing those edges by $pp'$ and $qq'$ (or by $pq'$ and $qp'$) has fewer crossings than $G$. Even more, the graph resulting from applying this transformation several times has fewer crossings than $G$.
\end{property}

\begin{corollary}\label{cor:quadrangular}
Let $S=R\cup B$ be in convex position with $|R|=|B|\geq 1$ or $|R|=|B|+1\geq 2$, and let $P$ be an oriented Hamiltonian alternating path on $S$ from $p$ to $q$. Suppose that $p_1q_1$ and $p_2q_2$ are two crossing edges of $P$ such that $q_1$ and $p_2$ have the same color. Then, there exists a Hamiltonian alternating path on $S$ from $p$ to $q$ with fewer crossings than $P$. See Figure~\ref{fig:quadrangular}.
\end{corollary}

\begin{theorem}\label{the:minimalpath}
Let $S=R\cup B$  be in convex position with $|R|=|B|\geq 1$ or $|R|=|B|+1\geq 2$, and let $p,q\in S$ such that the configuration $(S,p,q)$ is non-special. Then, every optimum Hamiltonian alternating path on $S$ with endpoints $p$ and $q$ is 1-plane.
\end{theorem}

\begin{proof}
Let $Opt$ be an optimum Hamiltonian alternating path with endpoints $p$ and $q$. We can consider $Opt$ as a directed path, where an edge $p_iq_j$ is directed from $p_i$ to $q_j$. Further, to apply symmetry, we complete the abstract graph $Opt$ to a cycle $C$ by adding the final (directed) edge $e=qp$. Edge $e$ is a dummy edge that is not crossed by any other edge (one can think on it as a curve from $p$ to $q$ placed outside $\ch(S)$), and its endpoints can have the same color.

Suppose on the contrary that $Opt$ is not 1-plane, i.e., there is at least one edge $e_1=p_1q_1$ in $Opt$ that is crossed by two other edges, say $e_2=p_2q_2$ and $e_3=p_3q_3$.

Assume, without loss of generality, that $C$ begins at point $p_1$ and then continues with: (1) edge $e_1$, (2) a path $P_1$ from $q_1$ to $p_2$, (3) edge $e_2$, (4) a path $P_2$ from $q_2$ to $p_3$, (5) edge $e_3$, (6) a path $P_3$ from $q_3$ to $p_1$. The dummy edge $e$ may belong to $P_1$, $P_2$ or $P_3$.

Suppose that point $p_1$ is red, and so $q_1$ is blue (the argument is analogous otherwise).
Since $e_1$ crosses edges $e_2, e_3$ and $Opt$ is optimum, by Corollary~\ref{cor:quadrangular}, points $p_2$ and $p_3$ are also red, and thus $q_2$ and $ q_3$ are blue. In particular, this implies that $P_2$ cannot consist of only one point.

The directed edge $e_1$ splits the points of $S\setminus \{p_1,q_1\}$ into two subsets, points to the left and points to the right of $e_1$, respectively. Assume, without lost of generality, that point $p_2$ is to the left of $e_1$ (otherwise, totally symmetric cases to those below would appear).

\begin{figure}[!htb]
\centering
\includegraphics[scale=0.7]{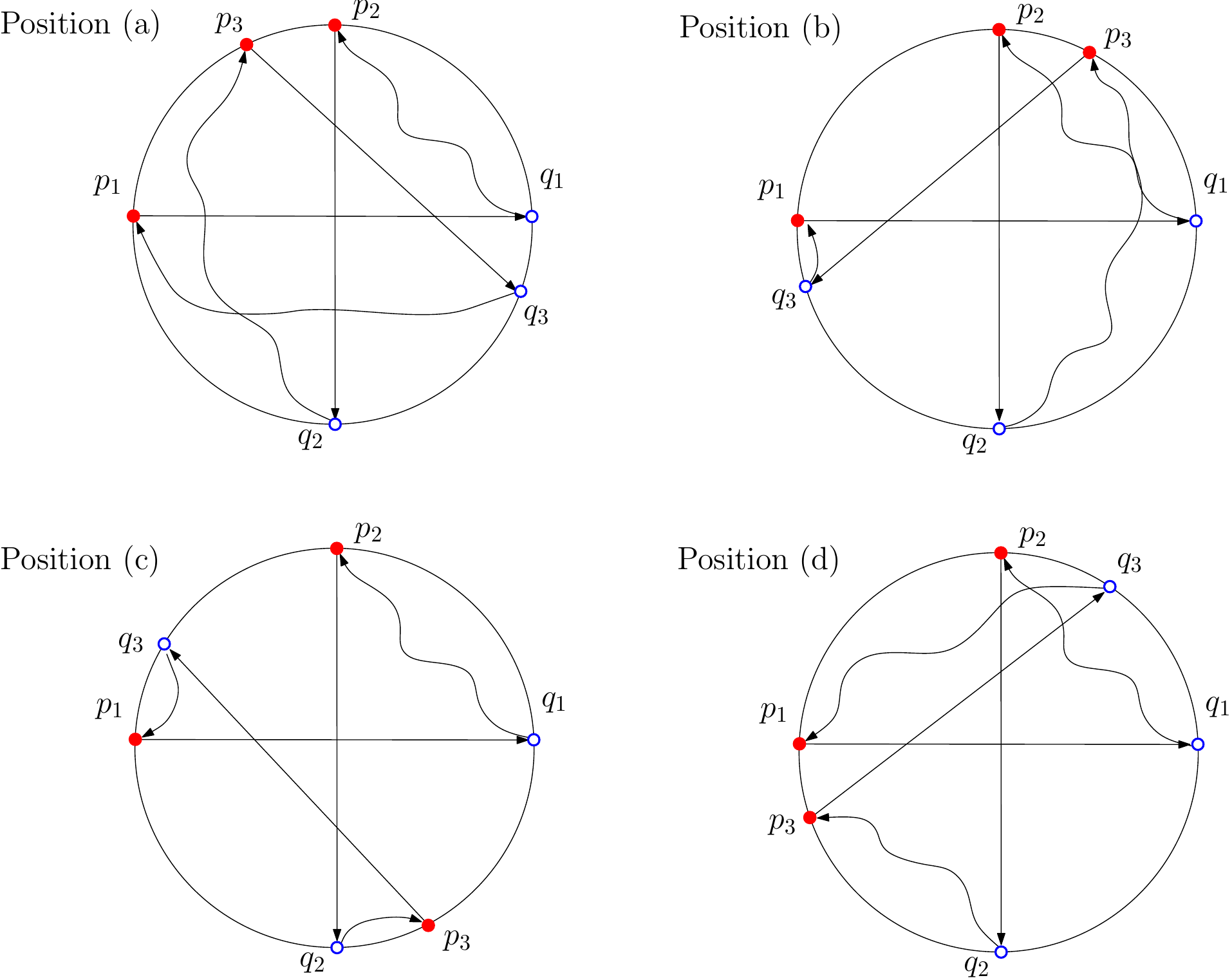}
\caption{The four positions when $e_2$ crosses $e_3$.}\label{fig:ThreeCrossings1}
\end{figure}

\begin{figure}[!tb]
\centering
\includegraphics[scale=0.7]{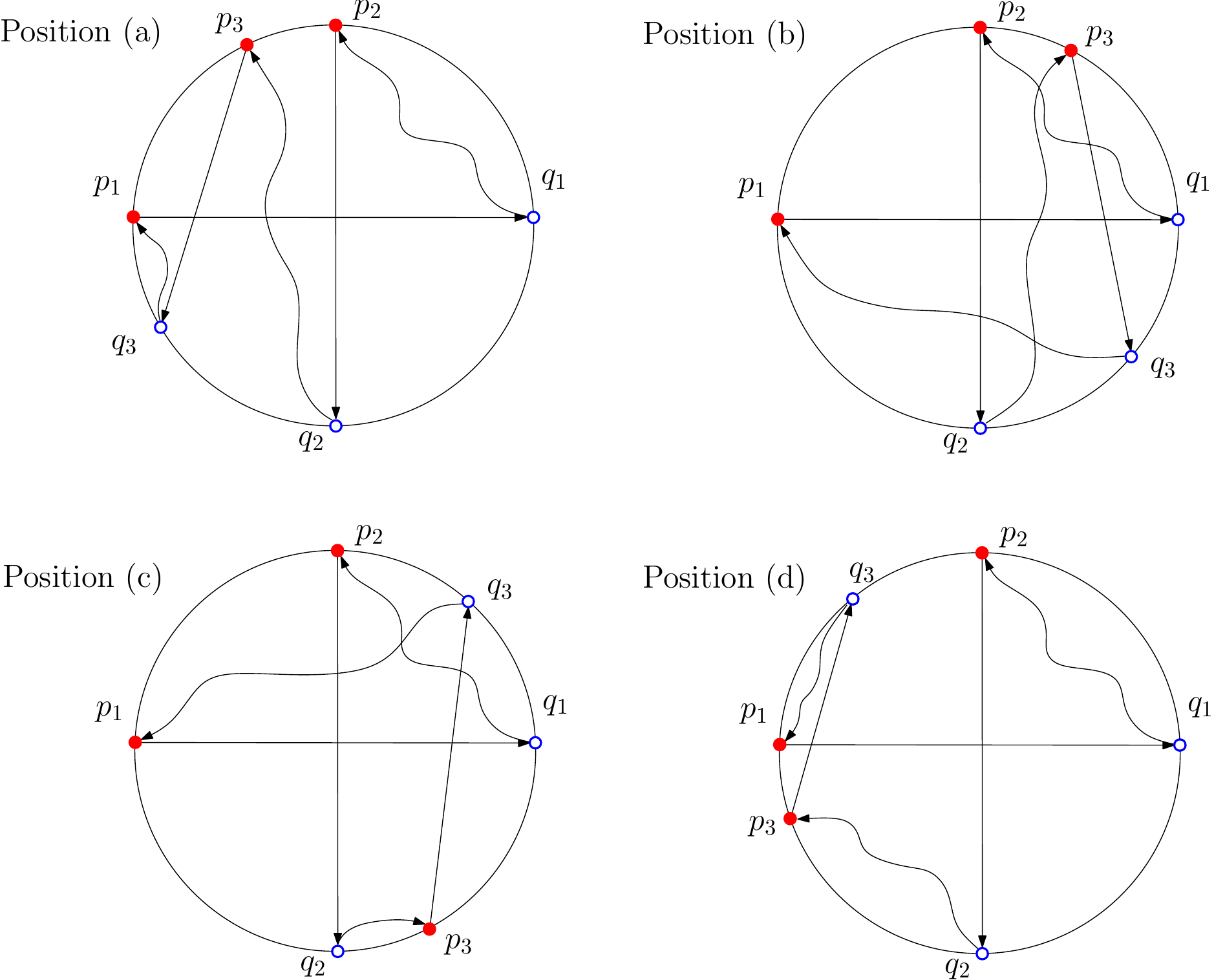}
\caption{The four positions when $e_2$ does not cross $e_3$.}\label{fig:TwoCrossings1}
\end{figure}

We next distinguish two cases depending on whether $e_2$ and $e_3$ intersect or not. Further, for each case, we shall consider four possibilities according to the position of $p_3$ with respect to $p_1, p_2, q_1$ and $q_2$ (clockwise).
\begin{itemize}
\item \emph{Position (a)}: $p_3$ is located between $p_1$ and $p_2$.
\item \emph{Position (b)}: $p_3$ is located between $p_2$ and $q_1$.
\item \emph{Position (c)}: $p_3$ is located between $q_1$ and $q_2$.
\item \emph{Position (d)}: $p_3$ is located between $q_2$ and $p_1$.
\end{itemize}

Figures~\ref{fig:ThreeCrossings1} and~\ref{fig:TwoCrossings1} show the four positions in each of the two cases. The three paths $P_1, P_2$ and $P_3$ are illustrated as curves, although in fact they are polygonal lines connecting points of $S$.

\vspace{0.3cm}
\noindent\emph{Case 1: edges $e_2$ and $e_3$ cross.}
\vspace{0.3cm}

In each of the four positions, by applying Property~\ref{pro:quadrangular} two or three times, we obtain a Hamiltonian alternating cycle that uses all the edges of the paths $P_1$, $P_2$ and $P_3$ and  has fewer crossings than $C$. This is a contradiction since, by removing edge $e$, there would be a path with fewer crossings than $Opt$. Therefore, $Opt$ cannot contain three edges crossing each other.

Figure~\ref{fig:ThreeCrossings2} shows the new cycles obtained for each position. The dotted edges are the original edges $e_1, e_2$ and $e_3$, the dashed edges represent intermediate edges that appear after applying Property~\ref{pro:quadrangular}, and the thick edges are the final edges belonging to the new cycle. Further, numbers 1, 2 and 3 denote the intersection points of the edges to which we apply Property~\ref{pro:quadrangular} and the order in which it is done. For example, for position (a), we first replace edges $p_2q_2$ and $p_3q_3$ by $p_2q_3$ and $p_3q_2$. Then, $p_1q_1$ and $p_3q_2$ are replaced by $p_1q_2$ and $p_3q_1$.

\begin{figure}[!htb]
\centering
\includegraphics[scale=0.7]{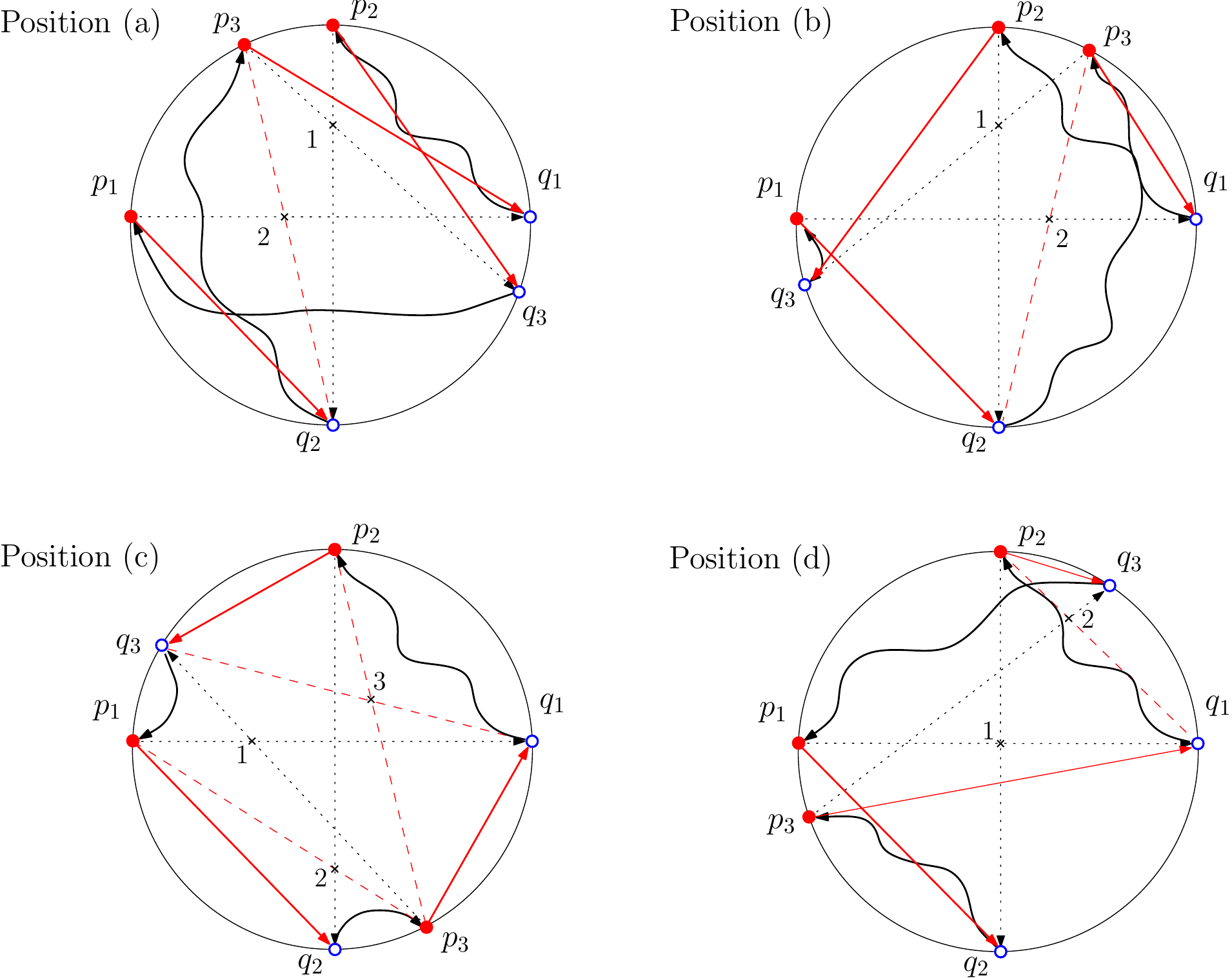}
\caption{Constructing new paths with fewer crossings than $Opt$ when $e_2$ crosses $e_3$.}\label{fig:ThreeCrossings2}
\end{figure}

\begin{figure}[!tb]
\centering
\includegraphics[scale=0.7]{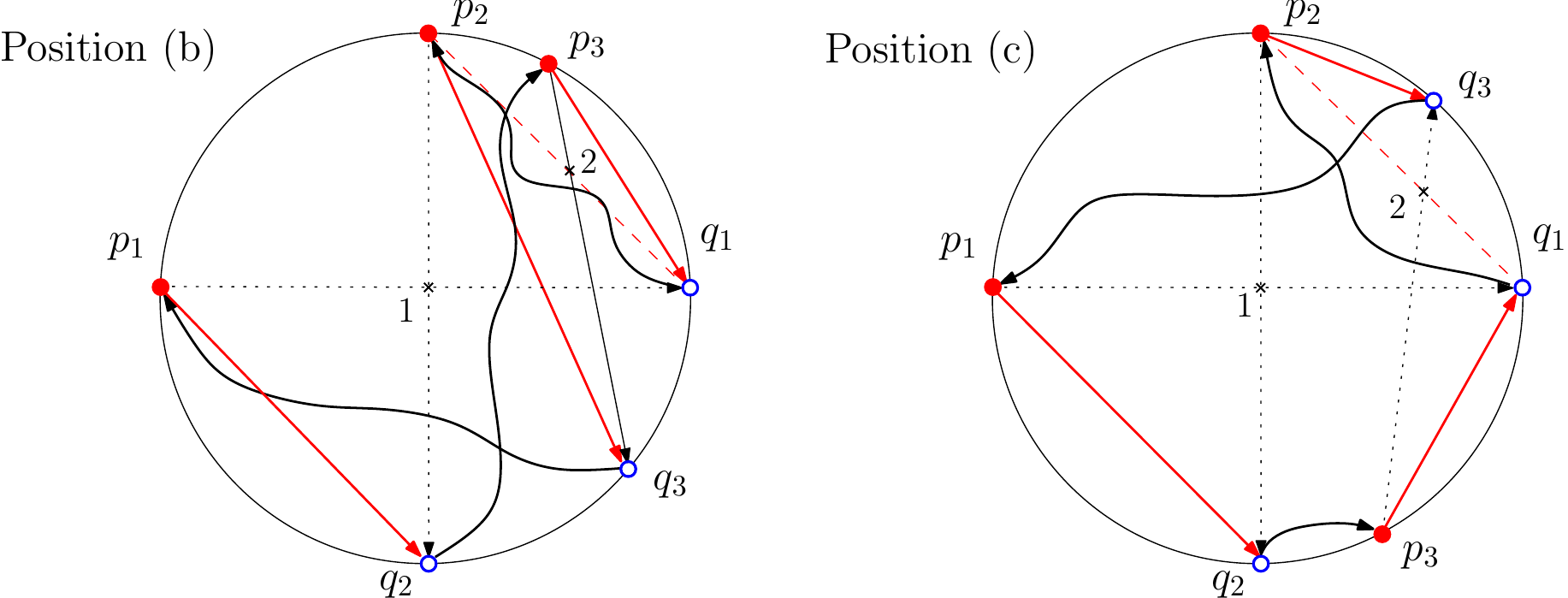}
\caption{Alternative paths (or cycles) for positions (b) and (c).}\label{fig:TwoCrossingsbc}
\end{figure}

\vspace{0.3cm}
\noindent\emph{Case 2: edges $e_2$ and $e_3$ do not cross.}
\vspace{0.3cm}

For positions (b) and (c) we reach the same contradiction as in case 1 by applying Property~\ref{pro:quadrangular} (twice); see Figure~\ref{fig:TwoCrossingsbc}. Thus, $Opt$ contains no edges with such positions for their endpoints.

\vspace{0.3cm}
\noindent
\emph{Case 2.1: The endpoints of $e_1,e_2,e_3$ are in position (d)}.
\vspace{0.3cm}

There can be different pairs of non-crossing edges intersecting $e_1$. Let us choose $e_2$ and $e_3$ such that the number of edges $k$ of $Opt$ with an endpoint between $q_3$ and $p_2$ (clockwise) and the other between $q_2$ and $p_3$ (again clockwise) is minimum.
Suppose that $k>0$, and let $e_d=p_dq_d$ be one of those $k$ edges. By Corollary~\ref{cor:quadrangular}, point $p_d$ is red and point $q_d$ is blue.

Assume first that point $p_d$ is to the left of edge $e_1$ (see Figure \ref{fig:TwoCrossingsd} left). Edge $e_d$ belongs to neither $P_1$ nor $P_3$; otherwise
$e_1,e_d$ and $e_2$ would be in position (b) (for $P_1$) or $e_1,e_3$ and $e_d$ would be in position (c) (for $P_3$). Further, if $e_d$ belongs to $P_2$ , then edges $e_1,e_d,e_3$ are also in position (d), but the number of edges of $Opt$ with one endpoint between $q_3$ and $p_d$, and the other  between $q_d$ and $q_3$ is smaller than $k$, which is a contradiction since $k$ is  minimum. Therefore, $p_d$ is to the right of edge $e_1$ which also leads to a contradiction in all the possible situations: $e_1,e_d, e_2$ are in position (c) (edge $e_d\in P_1$); $e_1,e_3,e_d$ are in position (b) (edge $e_d\in P_2$); $e_1,e_2,e_d$ are in position (d), and the contradiction comes again from the minimum value $k$. Note that, in fact, the two first situations are not exactly positions (b) and (c), but their symmetric ones.

Hence, if position (d) occurs in $Opt$ then $k=0$.

We now replace in $C$ edges $e_1,e_2,e_3$ by $p_1q_2$, $p_2q_3$, $p_3q_1$ (see Figure~\ref{fig:TwoCrossingsd} right), and thus obtain a new cycle $C'$.
Since there are no edges $e_d$ with an endpoint between $q_3$ and $p_2$ and the other between $q_2$ and $p_3$, one can check that an edge crossing one or two edges among $p_1q_2$, $p_2q_3$, $p_3q_1$ also crosses at least one or two of the replaced edges $e_1,e_2,e_3$. Hence, $C'\setminus \{e\}$ is a path that has fewer crossings than $Opt$, which is a contradiction.

Thus, we can conclude that $Opt$ contains no edges with endpoints in position (d).

\begin{figure}[tb]
\centering
\includegraphics[scale=0.7]{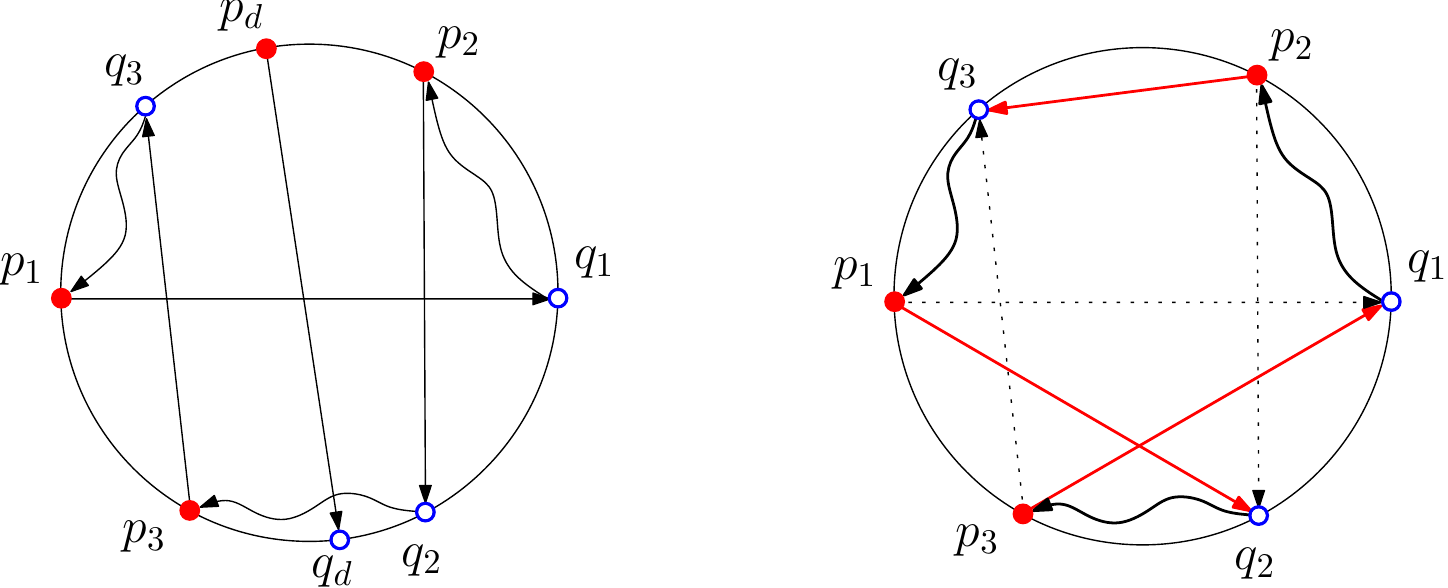}
\caption{Analysis of position (d).}\label{fig:TwoCrossingsd}
\end{figure}

\vspace{0.3cm}
\noindent\emph{Case 2.2: The endpoints of $e_1,e_2,e_3$ are in position (a)}.
\vspace{0.3cm}

Since positions (c) and (d) are forbidden, there is no edge in $Opt$ crossed by two edges with opposite directions. This property, which shall be called the {\it non-opposite edges property}, has several implications that will be assumed in the remaining of the proof:

\begin{itemize}
\item There cannot exist edges in $Opt$ with origin on the right of $e_1$ and extreme on its left.
\item The dummy edge $e=qp\in C$ does not belong to $Opt$ (and crosses no edge in $Opt$), and so it must be in path $P_2$ in order to go from $q_2$ to $p_3$, with $p$ to the left of $e_1$ and $q$ to its right.
\item All edges of $P_3$ must have both endpoints between $q_3$ and $p_1$ (clockwise); otherwise, either $e_3$ or $e_1$ would be crossed by edges in opposite directions. For the same reason, all edges of $P_1$ must have both endpoints between points $p_2$ and $q_1$ (clockwise).
\item  Point $p$ cannot be placed between $p_2$ and $q_1$ (clockwise), and $q$ cannot be a point between $q_3$ and $p_1$ (clockwise); otherwise, $e_2$ or $e_3$ would be crossed by edges in opposite directions.
\end{itemize}

An arbitrary straight-line intersecting $C$ and not containing points of $S$ is crossed by the same number of edges in direction left-right than in direction right-left. As a consequence, the difference between the number of edges crossing a given edge in one direction and in the opposite direction is at most 1. By the non-opposite edges property, an edge cannot be crossed by three different edges in $Opt$, but since one of the edges of $C$ is the dummy edge $e$, there can be an edge crossed by two edges (both in the same direction). Therefore, $e_1$ is intersected exactly twice, and $e_2$ and $e_3$ can be crossed once or twice.

Let $P'$ be the path from $p$ to $p_3$. We now distinguish four cases.

\vspace{0.3cm}
\noindent\emph{Case 2.2.1: Point $p$ is placed between $p_3$ and $p_2$, clockwise}. See Figure~\ref{fig:TwoCrossingsa} top left.
\vspace{0.3cm}

Both endpoints of all edges of path $P'$ are between $p_3$ and $p_2$ as $e_1$ is intersected only by edges $e_2$ and $e_3$, which cannot be crossed by two edges in opposite directions. Moreover, $P_3$ only visits points between $q_3$ and $p_1$, $P_1$ only visits points between $p_2$ and $q_1$, and the path from $q_2$ to $q$ cannot cross $e_1$. This implies that there are no points between $p_1$ and $p_3$.

Let $t$ be the previous point to $p_3$ in  $P'$. A new path from $p$ to $q$ is obtained by replacing edges $tp_3$, $e_1$, $e_2$ by $tp_2$, $q_1p_3$, $p_1q_2$, and removing edge $e$. Since there are no points between $p_1$ and $p_3$, and $e_1$ is only crossed by  $e_2$ and $e_3$,  an edge that intersects one of the new edges also crosses at least one of the removed edges (see Figure~\ref{fig:TwoCrossingsa} bottom left). Thus, the new path has fewer crossings than $Opt$, which is a contradiction.

\vspace{0.3cm}
\noindent\emph{Case 2.2.2: Point $p$ is placed between $p_1$ and $p_3$, clockwise, and no edge of  path $P'$ crosses edge $e_3$}. See Figure~\ref{fig:TwoCrossingsa} top center.
\vspace{0.3cm}

As in case 2.2.1, we again reach a contradiction by obtaining a new path from $p$ to $q$  with fewer crossings than $Opt$. To do this, it suffices to replace in $C$ edges $tp_3$, $e_1$ by $tp_1$, $p_3q_1$, and remove edge $e$ (see Figure~\ref{fig:TwoCrossingsa} bottom center).

\vspace{0.3cm}
\noindent\emph{Case 2.2.3: Point $p$ is placed between $p_1$ and $p_3$, clockwise, and exactly one edge $uv$ of path $P'$ crosses edge $e_3$}. See Figure~\ref{fig:TwoCrossingsa} top right.
\vspace{0.3cm}

By the non-opposite edges property,  $uv$ crosses $e_3$ from right to left. Further,  Corollary~\ref{cor:quadrangular} implies that $u$ is red and $v$ is blue. Thus, edges $uv$, $e_1$ and $e_2$ in $C$ can be replaced by $uq_1$, $p_2v$ and $p_1q_2$ to obtain a new path after removing the dummy edge $e$. Since $e_1$ is only intersected by $e_2$ and $e_3$, and $e_3$ is only crossed by $e_1$ and $uv$, one can easily check that this new path has fewer crossings than $Opt$ (see Figure~\ref{fig:TwoCrossingsa} bottom right), which is a contradiction.

\begin{figure}[tb]
\centering
\includegraphics[scale=0.7]{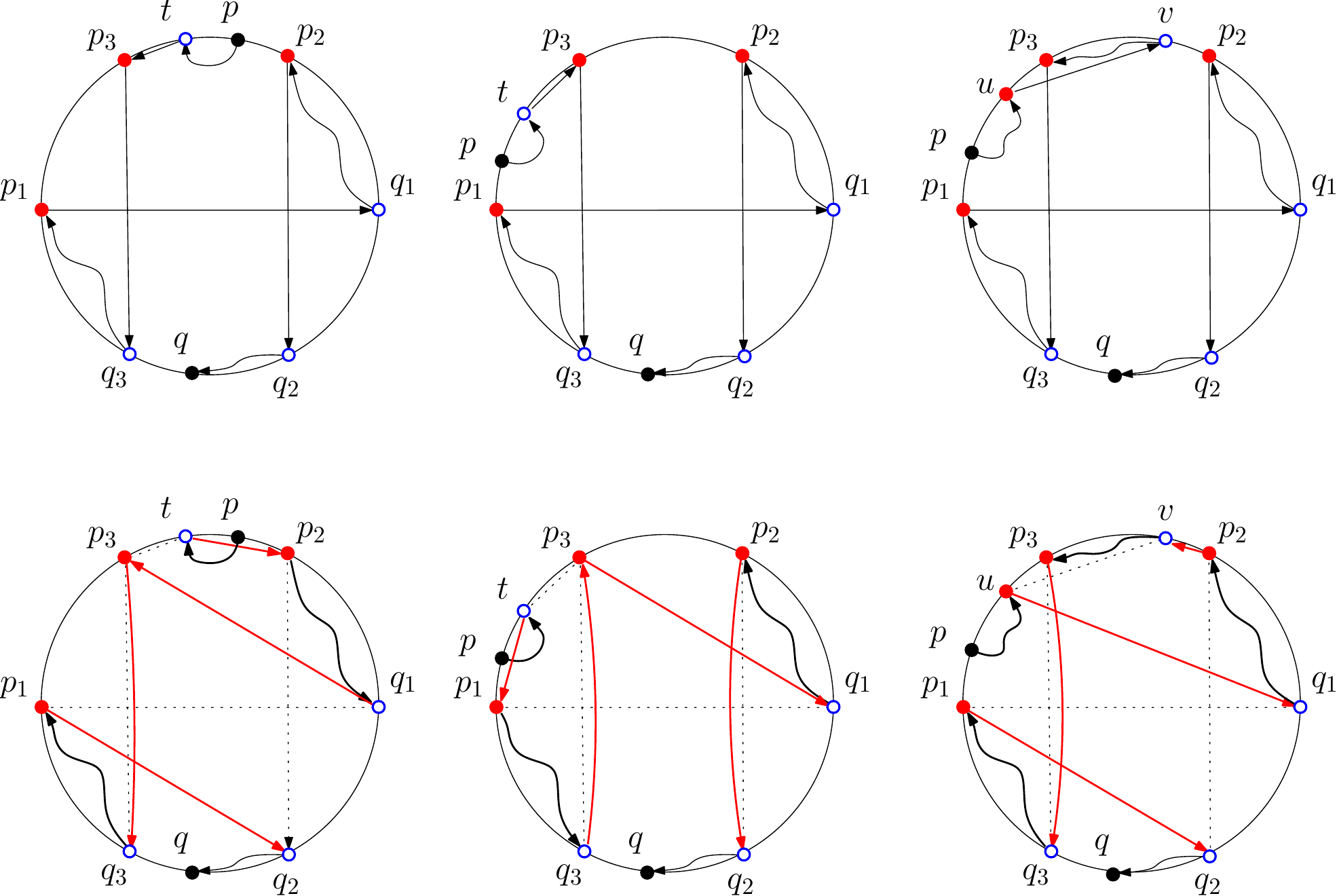}
\caption{Alternative paths for cases 2.2.1, 2.2.2 and 2.2.3.}\label{fig:TwoCrossingsa}
\end{figure}

\vspace{0.3cm}
\noindent\emph{Case 2.2.4: $p=p_3$}
\vspace{0.3cm}

By symmetry  and using $q$ instead of $p$, the same analysis as above can be done to show that $q$ must coincide with $q_2$. Thus, the two endpoints of $Opt$ have distinct color which implies that $|R|=|B|$ and, by hypothesis, the configuration $(S,p,q)$ is non-special.

Since $P_2$ consists of the dummy edge $e=qp$, path $P_3$ only visits points between $q_3$ and $p_1$ (clockwise), and $P_1$ only visits points between $p_2$ and $q_1$ (clockwise) it follows that $p=p_3$ is the unique point between $p_1$ and $p_2$, and $q=q_2$ is the unique point between $q_1$ and $q_3$. Thus, the neighbors of $p$ and $q$ on $S$ are both red and blue, respectively.

Let $R_1$ and $ B_1$ be, respectively, the sets of red and blue points that belong to path $P_1$, which has endpoints of different color, and so $|R_1|=|B_1|$. Further, this path only visits points between $p_2$ and $q_1$ (clockwise) which implies that it can be completed to a cycle $C_1$ (with edge $p_2q_1$) adding no new crossings.

By Theorem~\ref{the:convexcycles}, cycle $C_1$ (and hence $P_1$) has at least $|R_1|-\rr(R_1\cup B_1)$ crossings.
Analogously,  path $P_3$ has at least $|R_3|-\rr(R_3\cup B_3)$ crossings, where $R_3$ and $B_3$ are the sets of red and blue points that belong to that path. Therefore, $Opt$ has at least $$|R_1|-\rr(R_1\cup B_1)+|R_3|-\rr(R_3\cup B_3)+2$$ crossings (the two crossings in edge $e_1$ are also counted). Since $n=|R|=|R_1|+|R_2|+1$ (note that $R\setminus (R_1\cup R_2)=\{p\}$) and $\rr(R_1\cup B_1)+\rr(R_3\cup B_3)=\rr(S) +1$ (the red run to which $p$ belongs is counted twice), then the number of crossings in $Opt$ is at least $n-\rr(S)$.

Now, the configuration $(S,p,q)$ is non-special and the neighbors of $p$ and $q$ have their same respective color so there must be a partition $S_1\cup S_2$ of $S$, say the counterclockwise partition around $p$, such that $p_{i+1}\neq q$. Moreover, this partition is obtained by rotating a ray counterclockwise around $p$ and $S_1$ has minimum size, which implies that $p_{i+1}$ is located between $q_3$ and $p_1$ (clockwise). Then, we can argue as in case 2 of the proof of Lemma~\ref{lemm:path} to obtain a $1$-PHAP with at most $n-\rr(S)-1$ crossings: connect the paths $P_{S_1\cup\{p,p_{i+1}\}}(p,p_{i+1})$ and $P_{S_2}(p_{i+1},q)$ (see Figure~\ref{fig:consecutive} right). This is a contradiction since such a path has fewer crossings than $Opt$.

Therefore, in all cases we reach a contradiction and so the path $Opt$ is $1$-plane.
\end{proof}

For special configurations $(S,r,b)$, one cannot obtain  Hamiltonian alternating paths with endpoints $r$ and $b$ that are 1-plane (Theorem~\ref{the:specialconfig}), and thus, it is natural to ask about the minimum number of crossings that one might expect. The answer to this question is obtained as a consequence of Theorem~\ref{the:minimalpath}. Indeed, consider a special configuration $(S,r,b)$ and an optimum Hamiltonian alternating path $Opt$ with endpoints $r$ and $b$.
We can argue as in the proof of Theorem~\ref{the:minimalpath}, for all cases except the 2.2.4, to show that no edge of $Opt$ can be crossed twice. In case 2.2.4, the hypothesis of the configuration being non-special is only used in the last paragraph, and  the preceding construction can be used: $Opt$ has $n-\rr(S)$ crossings and consists of two optimum $1$-plane paths $P_1,P_3$ and one edge $e_1$, which is crossed twice. Further, by Theorem~\ref{Theorem}, $P_1$ and $P_3$ can be computed in linear time. Thus, we have the following corollary.

\begin{corollary}\label{the:special}
Let $(S,r,b)$ be a special configuration for a point set $S=R\cup B$ in convex position with $|R|=|B|=n\geq 1$. Then, every Hamiltonian alternating path on $S$ with endpoints $r$ and $b$ has at least $n-\rr(S)$ crossings. Moreover, an optimum Hamiltonian alternating path on $S$ with endpoints $r$ and $b$ can be computed in $O(n)$ time, has $n-\rr(S)$ crossings, and all its edges are crossed at most once, except for one of them that is crossed exactly twice.
\end{corollary}

Finally, we analyze the complexity of computing an optimum $1$-PHAP on  $S$ in convex position, provided that the endpoints of the path are fixed and the configuration is non-special (otherwise, by Theorem~\ref{the:specialconfig}, there does not exist such a path).

Let $<p_1,\ldots,p_{h}>$ be the clockwise circular order of the points of $S$. Recall the following notations from Section~\ref{subsec:cycles}; here, we consider $|R|=|B|=n\geq 1$ or $|R|=|B|+1=n+1\geq 2$.
 \begin{itemize}
\item sets $S[p_i,p_j]=\{p_i,p_{i+1}\ldots ,p_{j}\}$ (clockwise) and $S(p_i,p_j)=S[p_i,p_j]\setminus \{ p_i, p_j\}$.
 \item point $p_{J(i)}$ is the first point of $S$ (clockwise) such that $|S[p_i, p_{J(i)}]\cap R |=|S[p_i, p_{J(i)}]\cap B |$ for a given point $p_i$.
\item point $p_{J'(i)}$ is the first point of $S$ (counterclockwise) such that $|S[p_{J'(i)},p_i]\cap R |=|S[p_{J'(i)},p_i]\cap B |$ for a given point $p_i$.
\end{itemize}
Theorem \ref{the:quadratic} below describes an $O(n^2)$ time and space dynamic programming algorithm for computing an optimum $1$-PHAP on $S$ with fixed endpoints $p_s$ and $p_t$. As a main tool, it uses the following lemma. Observe that in both results the paths are considered to be oriented.

\begin{figure}[tb]
\centering
\includegraphics[scale=0.7]{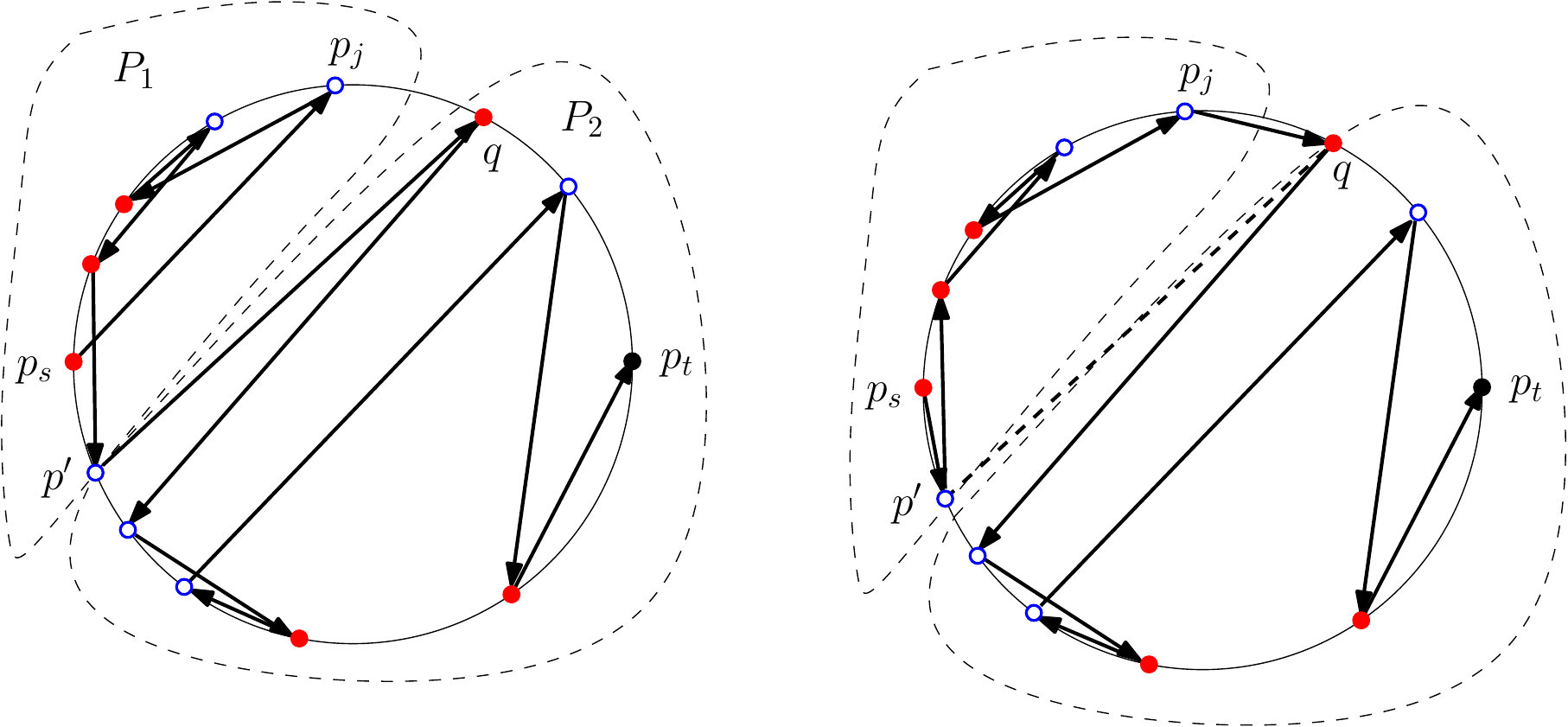}
\caption{Point $p'$ coincides with $p_{s-1}$}\label{AlgorithmOptimalPath1}
\end{figure}

\begin{lemma}\label{AuxiliarForPaths}
Let $S=R\cup B$ be in convex position with $|R|=|B|\geq 1$ or $|R|=|B|+1\geq 2$, and let $p_s,p_t$ be two non-consecutive points of $S$ such that the configuration $(S,p_s,p_t)$ is non-special if $p_s,p_t$ have distinct color. Then, there exists an optimum $1$-PHAP on $S$ from $p_s$ to $p_t$ that contains either an optimum sub-path from $p_s$ to $p_{J(s)}$, or an optimum sub-path from $p_s$ to $p_{J'(s)}$.
\end{lemma}

\begin{proof}
Consider an optimum $1$-PHAP, say $Opt$, directed from $p_s$ to $p_t$, and let $p_j$ be the point of $S$ visited after $p_s$ in $Opt$. Assuming that $p_j$ belongs to $S(p_s,p_{t})$, we show that there exists an optimum $1$-PHAP on $S$ from $p_s$ to $p_t$ that contains an optimum sub-path from $p_s$ to $p_{J(s)}$. A similar analysis can be done when $p_j$ belongs to $S(p_t,p_s)$, to show that there exists an optimum $1$-PHAP on $S$ from $p_s$ to $p_t$ that contains an optimum sub-path from $p_s$ to $p_{J'(s)}$.

If $p_j=p_{s+1}$, then $p_j$ coincides with $p_{J(s)}$. Otherwise, by Lemma~\ref{lemm:1convex}, path $Opt$ visits the following points in the given order: point $p_j$, all points in $\{p_{s+2}, \ldots , p_{j-1}\}$, point $p_{s+1}$, and  point $p'$, where $p'$ is either $p_{j+1}$ or $p_{s-1}$. Thus, $Opt$ consists of two optimum 1-plane paths that are connected at $p'$:
path $P_1$ from $p_s$ to $p'$ on $S[p_s,p_j]\cup \{p'\}$, and path $P_2$ from $p'$ to $p_t$ visiting the remaining points.

If $p_s$ and $p_{s+1}$ have different colors, then $p_s$ and $p'$ have the same color and are consecutive on $S[p_s,p_j]\cup \{p'\}$. By Corollary~\ref{cor:consecutive} $(ii)$, edge $p_sp_{s+1}$ must appear in $P_1$, and so $p_j=p_{s+1}$ which coincides with $p_{J(s)}$. Hence, we may assume that the color of $p_s$ is the same of $p_{s+1}$ and different of the $p'$.

Suppose that $p'=p_{s-1}$, and let $q$ be the point of $S$ visited  after $p'$ in $Opt$. A new path from $p_s$ to $p_t$ can be drawn by removing edges $p_sp_j$ and $p'q$, by adding edges $p_sp'$ and $p_jq$, and by reversing the orientation of the path from $p_j$ to $p_{s+1}$ (see Figure~\ref{AlgorithmOptimalPath1}). One can check that this new path has fewer crossings than $Opt$, a contradiction. Thus,  $p'\neq p_{s-1}$ which implies that $p'=p_{j+1}$.

Since the endpoints of $P_1$ have different colors, necessarily $S[p_{s},p_{j+1}]$ contains the same number of red and blue points, and so $p_{J(s)}\in S[p_{s},p_{j+1}]$. By Corollary \ref{cor:consecutive}$(i)$, path $P_1$ has $$\frac{|S[p_{s},p_{j+1}]|}{2}-\rr(S[p_s,p_{j+1}])$$ crossings. If $p_{J(s)} \ne p_{j+1}$, then $P_1$ can be replaced by a path that consists of two optimum 1-plane paths  connected at $p_{J(s)}$: path $P'$ on $S[p_{s},p_{J(s)}]$ from $p_s$ to $p_{J(s)}$, and path $P''$ on $S[p_{J(s)}, p_{j+1}]$ from $p_{J(s)}$ to $p_{j+1}$. By Corollary~\ref{cor:consecutive}, one can check that $P_1$ and $P'\cup P''$ have the same number of crossings.
\end{proof}

\begin{theorem}\label{the:quadratic}
Let $S=R\cup B$ be in convex position with $|R|=|B|\geq 1$ or $|R|=|B|+1\geq 2$, and let $p_s,p_t\in S$.
Then, an optimum Hamiltonian alternating path on $S$ from $p_s$ to $p_t$ can be computed in $O(n^2)$ time and space.
\end{theorem}

\begin{proof}
If $p_s$ and $p_t$ are consecutive points of $S$, by Theorem \ref{Theorem}, an optimum path from $p_s$ to $p_t$ can be computed in linear time. By Corollary \ref{the:special}, the same happens when $(S,p_s,p_t)$ is a special configuration. Thus, we may assume that $p_s$ and $p_t$ are non-consecutive, and that $(S,p_s,p_t)$ is non-special, and so, by Theorem \ref{the:minimalpath}, the optimum path is 1-plane.

By Lemma~\ref{AuxiliarForPaths}, an optimum path can be decomposed into two optimum paths: one from $p_s$ to $p_{J(s)}$ on $S[p_s,p_{J(s)}]$ (or from $p_s$ to $p_{J'(s)}$ on $S[p_{J'(s)},p_s]$) and the other from $p_{J(s)}$ to $p_t$ on $S[p_{J(s)}, p_{s-1}]$ (or from $p_{J'(s)}$ to $p_t$ on $S[p_{s+1}, p_{J'(s)}]$). By Theorem \ref{Theorem}, the first path can be computed in linear time.
Therefore, given $p_s\le p_i<p_t$ and $p_t\le p_j<p_s$, we have to compute an optimum $1$-PHAP on $S[p_i,p_j]$ with endpoints $p_t$ and either $p_i$ or $p_j$.

We explore all subsets $S[p_i,p_j]$ containing point $p_t$ (but not $p_s$) in increasing order of their size. Let $Cross_1[p_i,p_j]$ denote the number of crossings of an optimum $1$-PHAP on $S[p_i,p_j]$ from $p_i$ to $p_t$. Analogously, $Cross_2[p_i,p_j]$ is the number of crossings of an optimum $1$-PHAP on $S[p_i,p_j]$ from $p_j$ to $p_t$. Let $Cross[p_i,p_j]$ be the number of crossings of an optimum $1$-PHAP on $S[p_i,p_j]$ from $p_i$ to $p_j$. The value $Cross_1[p_i,p_j]$ can be computed by defining the following values $c_1$ and $c_2$:

$$c_1 = \left\{
  \begin{array}{ll}
    \infty & \hbox{if $p_{J(i)}\not\in S(p_i,p_t)$} \\
    Cross[p_i,p_{J(i)}]+Cross_1[p_{J(i)},p_j] & \hbox{if $p_{J(i)}\in S(p_i,p_t)$}
  \end{array}
\right.
$$

$$c_2 = \left\{
  \begin{array}{ll}
    \infty & \hbox{if $p_{J_1(i)}\not\in S[p_{t+1},p_j]$} \\
    Cross'[p_i,p_j]+Cross_2[p_{i+1},p_{J_1(j)}] & \hbox{if $p_{J_1(i)}\in S[p_{t+1},p_j]$}
  \end{array}
\right.
$$
where $p_{J_1(j)}$ is the counterclockwise first point after $p_j$ such that there are the same number of red and blue points in $\{p_i\}\cup S[p_{J_1(j)}, p_j]$, and $Cross'[p_i,p_j]$ is the number of crossings of an optimum $1$-PHAP from $p_i$ to $p_{J_1(j)}$ on $\{p_i\}\cup S[p_{J_1(j)},p_j]$ (see Figure~\ref{AlgorithmOptimalPath2} right).

Thus, by Lemma~\ref{AuxiliarForPaths}, $Cross_1[p_i,p_j]=\min\{c_1,c_2\}$. Note that, when $p_i$ and $p_t$ are consecutive,  by Corollary~\ref{cor:consecutive},  $Cross_1[p_i,p_j]$ is the difference between the red (blue) points and the red (blue) runs of $S[p_i,p_j]$.

We next show that indices $J(i), J_1(i)$ can be pre-computed in $O(n)$ time, which implies that $Cross[p_i,p_j]$ and $Cross'[p_i,p_j]$ can be computed in constant time.

\textsc{Procedure $J$-PAIRS} computes $p_{J(i)}$ in linear time. Moreover, if $ J_1(i)\ne j$, then $p_{J_1(i)}$ is the first counterclockwise point from $p_j$ such that the difference between the number of red and blue points in $S[p_{J_1(i)},p_j]$ equals one. It is not difficult to check that a slight modification of step 2 of \textsc{Procedure $J$-PAIRS} also computes $p_{J_1(i)}$ in linear time.

By numbering cyclically the runs of $S$ we can keep, for each point, the number of the run to which it belongs. Using this information, the number of red (blue) runs of any set $S[p_i,p_j]$ can be computed in constant time. In addition, by Corollary \ref{cor:consecutive}, $Cross[p_i,p_{J(i)}] = |S[p_i, p_{J(i)}]|/2 - \rr(S[p_i, p_{J(i)}])$, and so $Cross[p_i,p_{J(i)}]$ can  also be computed in constant time. The same reasoning applies to show that $Cross'[p_i,p_j]$ is computed in constant time.

A similar analysis can be done to compute $Cross_2[p_i,p_j]$. Therefore, each of the $|S(p_s,p_t)|\times |S(p_t,p_s)|$ values of $Cross_1[p_i,p_j]$ and $Cross_2[p_i,p_j]$ can be obtained in constant time. With these pre-computed data and using a standard dynamic programming algorithm, we can compute the number of crossings and the edges of an optimum $1$-PHAP on $S$ from $p_s$ to $p_t$ in $O(n^2)$ time and space.
\end{proof}

\begin{figure}[tb]
\centering
\includegraphics[scale=0.7]{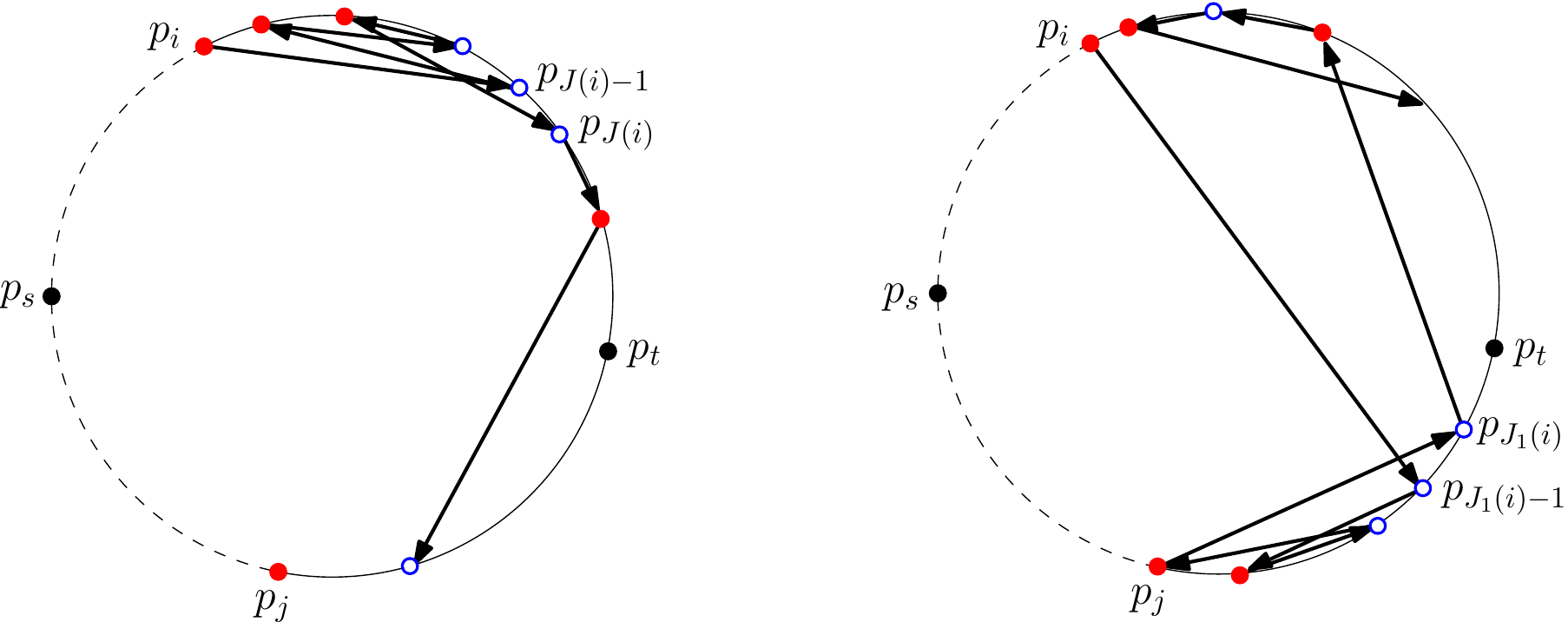}
\caption{The two cases to compute $Cross_1[p_i,p_j]$.}\label{AlgorithmOptimalPath2}
\end{figure}

To conclude this section, observe that $Cross_1[p_i, p_{i-1}]$ gives the minimum number of crossings of a Hamiltonian alternating path on $S$ from $p_i$ to $p_t$. Thus, the previous algorithm can be used to compute (in $O(n^2)$ time and space) all the optimum Hamiltonian alternating paths from any $p_i\neq p_t$ to $p_t$. Applying the algorithm $n$ times for any possible endpoint $p_t$, one can obtain the all-pairs optimum Hamiltonian alternating paths in $O(n^3)$ time.

\section{Open questions}\label{sec:conclu}

We conclude this paper by formulating a number of problems, related to those here considered, that remain open.

\begin{problem}
Characterize the special configurations $(S,r,b)$ that do admit a $1$-PHAP.
\end{problem}

Although for some point configurations $S$ in general position, it is not difficult to find a $1$-PHAC on $S$ with minimum number of crossings, in general, computing such a path seems to be an NP-hard problem. Thus, it would be interesting to solve the following two problems.

\begin{problem}
Among all Hamiltonian alternating cycles on $S$ with minimum number of crossings, determine whether there is at least one that is 1-plane.
\end{problem}

\begin{problem}
Decide whether it is possible to compute in polynomial time a $1$-PHAC on $S$ with minimum number of crossings.
\end{problem}

The construction of Remark~\ref{rem:converse}  works because the number of runs in $S$ is high. If $S$ contains only one red run and one blue run, then there is a unique $1$-PHAC on $S$. Thus, the following problem arises.

\begin{problem}
Prove that the number of 1-plane Hamiltonian alternating cycles on $S$ is exponential in $\rr(S)$ for $S$ in convex position.
\end{problem}

Finally, we would like to improve the $O(n^2)$ running time of the dynamic programming algorithm of Theorem~\ref{the:quadratic}, and so the following problem should be addressed.

\begin{problem}
Determine whether it is possible to reduce the number of explored subsets $S[p_i,p_j]$ to $o(n^2)$ (for example a linear subfamily, obtaining a linear time algorithm).
\end{problem}

\end{document}